\documentclass[runningheads]{llncs}
\usepackage[T1]{fontenc}
\usepackage{rotate}
\usepackage{multirow}
\usepackage{multicol}
\usepackage{comment}
\usepackage{bussproofs}
\usepackage{amssymb}
\usepackage{latexsym}
\usepackage{mathrsfs}
\usepackage{graphics}
\usepackage{fancyhdr}
\usepackage{comment}
\usepackage{tikz-cd}
\usepackage{proof}
\usepackage{color}
\usepackage[title]{appendix}
\usepackage[inline]{enumitem}

\usepackage{url}
\usepackage[linktoc=all,hidelinks,colorlinks=true,linkcolor=blue,citecolor=blue]{hyperref}

\usepackage{esvect}
\usepackage{orcidlink}

\usepackage{newtxmath}

\newenvironment{bprooftree}
  {\leavevmode\hbox\bgroup}
  {\DisplayProof\egroup}

\makeatletter
\def\namedlabel#1#2{\begingroup
    #2%
    \def\@currentlabel{#2}%
    \phantomsection\label{#1}\endgroup
}
\makeatother

\newcommand{\CLp}{\sf{CLp}}

\newcommand{\Pts}{\sf{PtS}}
\newcommand{\Bes}{\sf{BeS}}
\newcommand{\At}{\sf{At}}

\newcommand\ie{\hbox{\textit{i.e.}}}
\newcommand\eg{\hbox{\textit{e.g.}}}

\newcommand\dat{\Delta_{\At}}
\newcommand\gat{\Gamma_{\At}}
\newcommand\cfv{\Vdash^{\mathsf{cf}}}

\newcommand\bb{\mathcal{B}}
\newcommand\bc{\mathcal{C}}
\newcommand\bd{\mathcal{D}}
\newcommand\bu{\mathcal{U}}
\newcommand\bq{\mathcal{Q}}
\newcommand\st{\mathcal{ST}}
\newcommand\hs{\mathcal{HS}}

\newcommand\seq{\Rightarrow}

\newcommand{\ainit}{\mathsf{Ainit}}
\newcommand{\acut}{\mathsf{Acut}}

\newcommand{\init}{\mathsf{init}}

\newcommand{\ust}{\mathcal{U}_\st}
\newcommand{\uhs}{\mathcal{U}_\hs}

\begin{document}

\title{A Sequent Calculus Perspective on Base-Extension Semantics (Technical Report)}
\author{Victor Barroso-Nascimento\orcidlink{0000-0002-3990-5996} \and Ekaterina Piotrovskaya\orcidlink{0009-0009-4217-6948} \and Elaine Pimentel\orcidlink{0000-0002-7113-0801}}

\institute{Department of Computer Science, University College London, UK\\
\email{\{v.nascimento,kate.piotrovskaya.21,e.pimentel\}@ucl.ac.uk}}

\authorrunning{Barroso-Nascimento et al.}

\maketitle

\begin{abstract}

We define base-extension semantics ($\Bes$) using atomic systems based on sequent calculus rather than natural deduction. While traditional $\Bes$ aligns naturally with intuitionistic logic due to its constructive foundations, we show that sequent calculi with multiple conclusions yield a $\Bes$ framework more suited to classical semantics. The harmony in classical sequents leads to straightforward semantic clauses derived solely from right introduction rules.
This framework enables a Sandqvist-style completeness proof that extracts a sequent calculus proof from any valid semantic consequence. Moreover, we show that the inclusion or omission of atomic cut rules meaningfully affects the semantics, yet completeness holds in both cases.

\keywords{Base-extension Semantics  \and Proof-theoretic Semantics \and Sequent Calculus.}

\end{abstract}

\section{Introduction}
When proposing a new semantic model for a well-known logic, one must establish its adequacy in at least two ways: first, it should be correct and complete w.r.t. another characterisation of the logic, \eg\  a proof system, an axiomatic description, or another semantic model; second, it should exhibit some {\em novelty} or {\em advantage} compared to existing formulations.

For example, Kripke semantics is often favoured for counter-model extraction; truth tables are usually easy to compute and understand; and categorical semantics can help bridge logic with other areas since, at an abstract categorical level, different disciplines can end up looking quite similar. 

In this work, we will advocate for \emph{proof-theoretic semantics} as an adequate model for reasoning about {\em classical logics}.

Proof-theoretic semantics ($\Pts$) has recently attracted considerable attention, as it reduces {\em validity} to {\em derivability}, making it possible to employ the rich machinery of proof theory to describe the semantics of logical operators~\cite{Prawitz-2006,pts-91}. Beyond its technical appeal, $\Pts$ is grounded in {\em inferentialism}, offering a compelling philosophical account of the meaning of logical connectives, derived from the inference rules that govern them in a given proof system.

During the genesis of $\Pts$, it was often believed that its close association with natural deduction and apparent reliance  on intuitionistic proof-theoretic principles would yield completeness with respect to intuitionistic logic~\cite{Sanz2016-SANODV,PRAWITZ1971235}. Dummett even made the stronger claim that it would be impossible to obtain proof-theoretic semantics for classical logic without use of classical {\em canons of reasoning}~\cite[p. 270]{dummett1991logical}, which was later disproven by Sandqvist~\cite{Sandqvist} (see also~\cite{DBLP:journals/igpl/Makinson14,DBLP:journals/corr/abs-2503-05364}). Such beliefs were also justified by the very fact that the semantic clauses in the first days of $\Pts$ were inspired by introduction rules, since the possibility of defining validity through such means was interpreted as a consequence of the harmonic relations between rules that is characteristic of intuitionistic natural deduction~\cite{Dummett1973-DUMFPO-2}. 

According to traditional readings of the often quoted remark by Gentzen~\cite{Gentzen1969}, introduction rules can be interpreted as ``definitions'' of logical connectives~\cite{sep-proof-theoretic-semantics}. Dummett further argues that this definitional nature has semantic bearing due to making introduction rules self-justifying~\cite[p. 251]{dummett1991logical}.  On the other hand, since classical natural deduction requires intervention of extraneous principles instead of being fully determined by rules for the connectives, it would seem that such definitions are intuitionistic in nature. This is why, absent explicit inclusion of additional canons of reasoning, definitions based on introduction rules were specifically expected to yield intuitionistic semantics.

As remarked by Schroeder-Heister in Section 3.5. of~\cite{sep-proof-theoretic-semantics}, $\Pts$ seems to be biased in favour of intuitionistic logic {\em precisely} due to its reliance on natural deduction. Classical logic is perfectly harmonic if multiple succedent sequent calculus is considered instead~\cite{Cook10.1093/oxfordhb/9780195325928.003.0011}, arguably even more so than intuitionistic logic. Since validity definitions in $\Pts$ rely on proof-theoretic harmony to provide semantics for the logical connectives, one could reasonably expect a sequent-based definition of $\Pts$ to yield a semantics for classical logic.

This paper fulfils this expectation by breaking with the $\Pts$ tradition of relying on natural deduction and intuitionistic systems and introducing the first Base-extension semantics ($\Bes$) formulated for sequent systems. 

$\Bes$ is a strand of $\Pts$ in which proof-theoretic validity is defined relative to a specified set of inference rules governing basic formulas of the language~\cite{Sandqvist2015IL}. These basic formulas are often atomic propositions such as ``$a$ is a subset of $b$'', ``$b$ is a subset of $a$'', or ``$a$ is equal to $b$''. If we represent these as $r$, $s$, and $t$, respectively, a possible inference rule might be:
\[
\infer{t}{r & s}
\]
This resembles the approach of atom definitions~\cite{DBLP:journals/logcom/HallnasS91,DBLP:journals/tocl/McDowellM02,DBLP:journals/tocl/MillerT05,DBLP:journals/apal/MarinMPV22}, where, for instance, set equality can be ``defined'' in terms of mutual containment---we note, however, that this analogy is purely motivational; we are not claiming to define set-theoretic equality within propositional logic.

The notion of consequence in $\Bes$ is given by an inductively defined semantic judgment, whose base case refers to provability in a given atomic system (or a \emph{base}). That is, for any base $\bb$, provability of an atomic formula is defined as:
\[
\Vdash_{\bb} p \mbox{ iff } p \mbox{ is provable in } \bb 
\]
The remaining semantic clauses are defined inductively. For example, conjunction is handled similarly to Kripke semantics:
\[\Vdash_{\bb} A \land B \mbox{ iff  } \Vdash_{\bb} A \mbox{ and} \Vdash_{\bb} B\]
Beyond this point, the similarity to Kripke semantics fades. For instance, the constant $\bot$ is often defined as {\em absurdity} in second order logic instantiated to atoms (as in the system $\mathsf{Fat}$~\cite{DBLP:journals/jsyml/FerreiraF13}):
\[\Vdash_{\bb} \bot \mbox{ iff  } \Vdash_{\bb} p \mbox{ for all atoms} \ p \]
The use of semantic clauses making $\bot$ unsatisfiable, as in Kripke semantics, leads to serious issues, as it allows one to prove $\Vdash_{\mathcal{B}}  \neg \neg p$ for every atom $p$~\cite{DBLP:journals/jphil/PiechaSS15}. It is still possible, however, to avoid this problem by considering $\bot$ as a basic formula and requiring bases to be consistent (see \eg\ \cite{DBLP:journals/corr/abs-2306-03656,barrosonascimento2025prooftheoreticapproachsemanticsclassical}). 

Shifting from natural deduction to sequent systems brings a whole new perspective to the meaning of connectives---it also introduces new challenges, as discussed in the next section. Indeed, in classical logic, $\bot$ is introduced from the empty sequent via an invertible rule, and our semantic definition adequately reflects this behaviour:
\[
\Vdash_{\mathcal{B}} \bot, \Gamma \mbox{ iff }\Vdash_{\mathcal{B}} \Gamma
\]
This allows the definition of connectives using the invertibility of rules, aligning the semantic clauses closely with the proof rules of classical sequent systems, as well as bringing {\em proof theory meta-reasoning} to the core of $\Pts$.

Moreover, our formulation also supports a constructive proof of completeness for classical logic, mirroring the strategy employed by Sandqvist to prove completeness of traditional $\Bes$ with respect to intuitionistic logic~\cite{Sandqvist2015IL}. In this setting, quantification over atoms, previously used in the semantic clauses for disjunction and $\bot$, is no longer necessary.

Finally, the sequent-style formulation of $\Bes$ deepens our understanding of general $\Pts$. In particular, it makes explicit the semantic role of the cut rule, an aspect often obscured in natural deduction-based approaches.

\section{Sequent based base-extension semantics}\label{sec:bes}
$\Bes$ is founded on an inductively defined judgment called {\em support}, which mirrors the syntactic structure of formulas. The inductive definition begins with a base case: the support of atomic propositions is determined by derivability in a given {\em base}---a specified collection of inference rules that govern atomic propositions. 


In this work, we adapt Sandqvist's terminology~\cite{Sandqvist2015IL} to the sequent setting.

\subsection{Atomic Derivability}

The $\Bes$ begins by defining derivability in a {\em base}. We use, as does Sandqvist, systems containing rules over basic sentences for the semantical analysis.


The propositional {\em base language} is assumed to have a set $\At=\{p_1, p_2,\ldots\}$ of countably many atomic propositions.
The elements of $\At$ will be called {\em basic sentences}, or simply {\em atoms}.

Let $\Gamma_{\At},\Delta_{\At}$ be 
sets of atoms. An {\em atomic sequent} is denoted by $\gat \Rightarrow \dat$, where  $ \Rightarrow \dat$ denotes  $\varnothing \Rightarrow \dat$.  

\begin{definition}[Base]
An {\em atomic sequent system} (or {\em sequent base}) $\bb$ is a (possibly empty) set of atomic sequent rules of the form

\begin{center}

\begin{bprooftree}
    \AxiomC{$\Gamma^{1}_{\At} \Rightarrow \dat^{1}$}
    \AxiomC{\ldots}
     \AxiomC{$\Gamma^{n}_{\At} \Rightarrow \dat^{n}$}
     \TrinaryInfC{$\gat \Rightarrow \dat$}
\end{bprooftree}

\end{center}
The sequence $\langle \Gamma^{1}_{\At} \Rightarrow \dat^{1}, \ldots , \Gamma^{n}_{\At} \Rightarrow \dat^{n} \rangle$ of premises in a rule can be empty---in this case the rule is called an {\em atomic axiom}.
\end{definition}
As in this work we will be referring to sequent systems only, from now on we will drop the word ``sequents'' from the objects just defined.

\begin{definition}[Extension]
    Let $\bb,\bc$ be bases. We say that $\bc$ is an {\em extension} of  $\bb$ (written $\bc \supseteq \bb$) if $\bb$ is a subset of $\bc$.
\end{definition}

Before defining atomic derivability, it is important to clarify key distinctions between natural deduction and sequent calculus. 

First, sequent calculus derivations always proceed from unconditional tautologies to unconditional tautologies, whereas natural deduction relies on hypothetical assumptions, often treating contexts implicitly. Second, the rules in sequent calculus must explicitly account for contexts---even when arbitrary---while natural deduction may treat contexts implicitly or explicitly.

These differences carry subtle but important consequences. For instance, the identity axiom $A \Rightarrow A$ must be explicitly included in sequent systems, whereas in natural deduction such derivability is implicit in the presence of an assumption $A$. 
Similarly, consider the contrast between the following rules:

\begin{prooftree}
\AxiomC{$A \vee B$}
\AxiomC{$[A]$}
\noLine
\UnaryInfC{$\vdots$}
\noLine
\UnaryInfC{$C$}
\AxiomC{$[B]$}
\noLine
\UnaryInfC{$\vdots$}
\noLine
\UnaryInfC{$C$}
\RightLabel{\scriptsize{$\vee$-elim}}
\TrinaryInfC{$C$}
\DisplayProof
\qquad
\AxiomC{}
\noLine
\UnaryInfC{}
\noLine
\UnaryInfC{}
\noLine
\UnaryInfC{}
\noLine
\UnaryInfC{}
\noLine
\UnaryInfC{}
\noLine
\UnaryInfC{}
\noLine
\UnaryInfC{$\Gamma, A \Rightarrow \Delta$}
\AxiomC{}
\noLine
\UnaryInfC{}
\noLine
\UnaryInfC{}
\noLine
\UnaryInfC{}
\noLine
\UnaryInfC{}
\noLine
\UnaryInfC{}
\noLine
\UnaryInfC{}
\noLine
\UnaryInfC{$\Gamma', B \Rightarrow \Delta'$}
\RightLabel{$L \lor$}
\BinaryInfC{$\Gamma, \Gamma', A \lor B \Rightarrow \Delta, \Delta'$}
\end{prooftree}

It is implicitly assumed that $\lor$-elim can be applied regardless of the sets of formulas $\Delta$, $\Delta'$ and $\Delta''$ used to derive its major premise $A \lor B$ and minor premises $C$, but the rule $L \lor$ represents such contextual sets explicitly (albeit still letting them be arbitrary). It comes as no surprise, then, that natural deduction systems for logics in which context matters (such as linear logic) have to make them explicit again~\cite{LinearLogicND10.1093/jigpal/12.6.601}. 

This distinction becomes critical in atomic systems. In natural deduction, atomic rules can be applied freely, independent of the surrounding context, as illustrated by the rule on the left below. 
\[
\infer{r}{\deduce{t}{\deduce{\vdots}{[p,q]}}} \qquad \qquad
\infer{\seq r}{p,q\seq t}
\]
In contrast, a sequent rule---such as the one on the right above---can only be applied when the exact context 
$p,q$ is present\footnote{This in a premises-to-conclusion reading of the rule, as done in proof-construction. In the bottom-up reading, as done in proof-search, the context should be empty in order to apply this rule.}. This highlights the need to decide how contexts are handled in sequent-style atomic calculi.

If contexts are not treated as arbitrary, three main difficulties arise:
\begin{itemize}
\item {\em Proof rules:} To simulate generality, one rule must be included for every possible atomic context. For example, consider a base that contains a rule of the following form for each set $\Gamma_{\At}$ of atomic formulas:

\begin{prooftree}
\AxiomC{$\Gamma_{\At}, q \Rightarrow t$}
\UnaryInfC{$\Gamma_{\At} \Rightarrow r$}
\end{prooftree}

\noindent
Even if contexts are syntactically fixed, the inclusion of a rule for every conceivable atomic context allows the base to behave as though it included a single rule defined over arbitrary contexts. While this is technically feasible---\eg, through infinitary rule bases---it complicates the system (see~\cite{barrosonascimento2025prooftheoreticapproachsemanticsclassical}).

\item {\em Compositionality:} Fixed contexts restrict derivation composition to cases where sequents match exactly, limiting expressiveness:

\begin{prooftree}
    \AxiomC{$\Gamma_{\At}  \Rightarrow p$}
    \AxiomC{$\Delta_{\At}, p \Rightarrow q$ }
    \RightLabel{\scriptsize{Cut}}
    \BinaryInfC{$\Gamma_{\At}, \Delta_{\At} \Rightarrow q$}
\end{prooftree}

\noindent
To address this, one can close bases w.r.t. atomic cut rules, which allow compositionality through controlled context merging. We will show that adopting this choice will have a direct impact on the semantics principles that can be validated.

\item {\em Context cumulativity:} Consider the following atomic rule:

\begin{prooftree}
    \AxiomC{$\Gamma_{\At} \Rightarrow p$}
    \AxiomC{$\Delta_{\At}, q \Rightarrow r$}
    \BinaryInfC{$\Theta_{\At} \Rightarrow s$}
\end{prooftree}

\noindent
Even though contexts are always carried from premises to conclusion in both natural deduction and sequent calculus, sets of atoms occurring in the premises are not required to be related in any way to the sets occurring in the conclusion if contexts are fixed.
This means that in fixed-context settings context accumulation must be enforced at the level of the rules themselves, rather than being built into the notion of deduction. This allows for simpler rule definitions, but undermines general guarantees like context monotonicity across base extensions. Again, atomic cut rules can reintroduce cumulativity, but at the cost of complicating the calculus.
\end{itemize}
Given these challenges, and motivated by the semantic role of cut in proof-theoretic semantics, we choose to allow arbitrary atomic contexts in our definition of derivability. This ensures generality, supports context cumulativity, and facilitates semantic interpretation. We also adopt an unconditional notion of derivability, in line with sequent calculus traditions, though a hypothetical variant could be used with similar results.

\begin{definition}[Derivability] \label{def:derivability}
For every $\mathcal{B}$, the \emph{derivability relation} (written $\vdash_{\mathcal{B}}$) is defined as follows.

\begin{description}[itemsep=0.5em, font=\normalfont]



\item[(Axiom/Weakening)] For any atomic axiom in $\bb$ of the form  \[\infer{\Gamma_{\At} \Rightarrow \Delta_{\At}}{}\]  $\vdash_{\bb} \Theta_{\At}, \Gamma_{\At} \Rightarrow \Delta_{\At}, \Sigma_{\At}$ holds for any sets $\Theta_{\At}$ and $\Sigma_{\At}$ of atoms;

\item[(Mix)] If $\vdash_{\bb}\Theta^{1}_{\At}, \Gamma^{1}_{\At}  \Rightarrow \Delta^{1}_{\At} , \Sigma^{1}_{\At}, \ \  \ldots, \ \  \vdash_{\mathcal{B}}\Theta^{n}_{\At} , \Gamma^{n}_{\At} \Rightarrow \Delta^{n}_{\At} , \Sigma^{n}_{\At}$ hold in $\bb$
and the following rule is in $\bb$:
\[
\infer{\Gamma_{\At} \Rightarrow \Delta_{\At}}{\Gamma^{1}_{\At} \Rightarrow \Delta^{1}_{\At}&\ldots&\Gamma^{n}_{\At} \Rightarrow \Delta^{n}_{\At}}
\]
 then $\vdash_{\bb}\Theta^{1}_{\At} , \ldots , \Theta^{n}_{\At} , \Gamma_{\At} \Rightarrow \Delta_{\At} , \Sigma^{1}_{\At} , \ldots , \Sigma^{n}_{\At}$ holds.

\end{description}

\end{definition}

The following rules play a central role in our framework:

\begin{prooftree}
\AxiomC{}
\RightLabel{\small $\ainit$}
\UnaryInfC{$\Gamma_{\At}, p \Rightarrow p, \Delta_{\At}$}
\DisplayProof
\qquad
    \AxiomC{$\Gamma_{\At}^1 \Rightarrow \dat^1,p$}
    \AxiomC{$p, \gat^2 \Rightarrow \dat^2$}
    \RightLabel{\small $\acut$}
    \BinaryInfC{$\Gamma_{\At}^1, \gat^2 \Rightarrow \dat^1,\dat^2$ }
\end{prooftree}

\noindent
Indeed, we will restrict our attention to systems that are closed under the $\ainit$ rule for every atom~$p$ and all sets of atoms~$\Gamma_{\At}, \Delta_{\At}$. We observe that, although the contexts $\Gamma_{\At}$ and $\Delta_{\At}$ in $\ainit$ are not strictly necessary, they simplify certain steps in the completeness proof.

Regarding $\acut$, we will consider both: systems that are closed under this rule and those that are not.
\begin{definition}[Structural/Cut-free base] \label{def:stbase}
  The structural (cut-free) base $\mathcal{ST}$ ($\mathcal{HS}$) is the smallest atomic system closed under atomic axiom and atomic cut (atomic axiom).
\end{definition}

As will be shown later, even though classical logic is complete with respect to our semantics regardless of the presence or absence of atomic cuts, the cut rule proves to be not only impactful for the semantics of individual bases, but also relevant for determining the shape of the completeness proof.

\subsection{Semantics}
In this subsection we define the proposed semantics and sketch some results that will be key for the soundness and completeness results. 

We start with the definition of the support relation, which reduces to derivability in $\bb$ in the base case.

\begin{definition}[Support] \label{def:support}
Let $X_{\At}\subseteq\At$ for any set $X$.
    \emph{Support} in a base $\bb$ (written $\Vdash_{\bb}$) is defined as follows:

    \begin{itemize}[itemsep=0.5em, font=\normalfont]
        \item[\namedlabel{eq:supp-at}{(At)}] $\Vdash_{\bb} \Gamma_{\At}$ iff $\vdash_{\bb} \ \Rightarrow \Gamma_{\At}$;
        
        \item[\namedlabel{eq:supp-and}{$(\land)$}] $\Vdash_{\bb} A \land B, \Gamma$ iff $\Vdash_{\bb} A, \Gamma$ and $\Vdash_{\bb} B, \Gamma$;
        
        \item[\namedlabel{eq:supp-or}{$(\lor)$}] $\Vdash_{\bb} A \lor B, \Gamma$ iff $\Vdash_{\bb} A, B, \Gamma$;
        
        \item[\namedlabel{eq:supp-imp}{$(\to)$}] $\Vdash_{\bb} A \to B, \Gamma$ iff $A \Vdash_{\bb} B, \Gamma$;

        \item[\namedlabel{eq:supp-bot}{$(\bot)$}] $\Vdash_{\mathcal{B}} \bot, \Gamma$ iff $\Vdash_{\mathcal{B}} \Gamma$
        
        \item[\namedlabel{eq:supp-inf}{(Inf)}] $\left\{A^{1},\ldots,A^{n}\right\} 
        \Vdash_{\bb} \Delta$ iff for all $\bc \supseteq \bb$ and for all $\left\{\Theta_{\At}^{i}\right\}_{1\leq i\leq n}$, if $\Vdash_{\bc} \Theta^{i}_{\At},A^{i}$ for all $1\leq i\leq n$ 
        then $\Vdash_{\bc} \Theta^{1}_{\At}, \ldots , \Theta^{n}_{\At}, \Delta$
    \end{itemize}
\end{definition}

\begin{definition}[Validity/Cut-free validity] \label{def:validity}
    We say that an inference from $\Gamma$ to $\Delta$ is \emph{valid} (\emph{cut-free valid}), written as $\Gamma \Vdash \Delta$ ($\Gamma \cfv\Delta$), if $\Gamma \Vdash_{\bb} \Delta$ for all $\bb \supseteq \st$ ($\bb \supseteq \hs$).
\end{definition}

Observe that the $\Bes$ approach for the semantic clause for disjunction in intuitionistic logic presented in~\cite{Sandqvist2015IL} closely mirrors the disjunction elimination rule in natural deduction: 
\[
\Vdash_{\bb} A \lor B \mbox{ iff }  \forall \bc \supseteq \bb,p \in {\At}, \mbox{ if } A \Vdash_{\bc} p \mbox{ and } B \Vdash_{\bc} p \mbox{ then } \Vdash_{\bc} p
\]
In particular, note the need for the use of quantification over atoms for the definition of semantic clauses. 

This is not the case in the sequent presentation of $\Bes$, which is perfectly capable of using only right rules (the sequent calculus equivalent of natural deduction's introduction rules) for the definition of its clauses. In particular, the semantic clause for disjunction reflects the invertible multiplicative right inference rule in sequent calculus (see Fig~\ref{fig:CL}).

%

We continue with the standard monotonicity result for bases, which is proven by induction on the complexity of context formulas.

\begin{lemma}[Monotonicity] \label{lm:monotonicity}
    If $\Vdash_{\bb} \Delta$ and $\bc \supseteq \bb$, then $\Vdash_{\bc} \Delta$.
\end{lemma}
\begin{proof}
We will prove, by induction on the complexity of formulas in $\Delta$, that if $\Vdash_{\bb} \Delta$ and $\bc \supseteq \bb$, then $\Vdash_{\bc} \Delta$.

    Base case: $\Delta$ is atomic, {\ie} $\Delta = \Delta_{\At}$; if $\vdash_\bb \Rightarrow \Delta_{\At}$ then $\vdash_\bc \Rightarrow \Delta_{\At}$ (since $\mathcal{C}$ contain all rules of $\mathcal{B}$), and so if $\Vdash_{\bb} \Delta_{\At}$ then $\Vdash_{\bc} \Delta_{\At}$ by definition.
   
    Inductive case: as the induction hypothesis, take that $\Vdash_{\bc} \Delta$ holds true for any proper subset of $\Delta$ in place of $\Delta$.
    \begin{description}[itemsep=0.2em]
        \item[$\Delta = A \land B, \Delta':$] assume that $\Vdash_{\bb} A \land B, \Delta'$. Then $\Vdash_{\bb} A, \Delta'$ and $\Vdash_{\bb} B, \Delta'$ by~\ref{eq:supp-and}. Then $\Vdash_{\bc} A, \Delta'$ and $\Vdash_{\bc} B, \Delta'$ by induction hypothesis; by~\ref{eq:supp-and}, $\Vdash_{\bc} A \land B, \Delta'$.
        
        \item[$\Delta = A \lor B, \Delta':$] assume that $\Vdash_{\bb} A \lor B, \Delta'$. Therefore $\Vdash_{\bb} A, B, \Delta'$ by~\ref{eq:supp-or}. Then $\Vdash_{\bc} A, B, \Delta'$ by induction hypothesis, hence $\Vdash_{\bc} A \lor B, \Delta'$ by~\ref{eq:supp-or}.
        
        \item[$\Delta = A \to B, \Delta':$] assume that $\Vdash_{\bb} A \to B, \Delta'$. Then $A \Vdash_{\bb} B, \Delta'$ by~\ref{eq:supp-imp}. Let $\Theta_{\At}$ be an arbitrary atomic context. Take an arbitrary $\bd \supseteq \bc \supseteq \bb$ such that $\Vdash_{\bd} \Theta_{\At}, A$. Notice that, by transitivity of base extensions, $\bd \supseteq \bb$. Since $A \Vdash_{\bb} B, \Delta$ and $\Vdash_{\bd} \Theta_{\At}, A$, by~\ref{eq:supp-inf} we have $\Vdash_{\bd} \Theta_{\At}, B, \Delta$. Since $\mathcal{D} \supseteq \bc$ such that $\Vdash_{\bd} \Theta_{\At}, A$, we conclude $A \Vdash_{\bc} B, \Delta$ by~\ref{eq:supp-inf}, hence $\Vdash_{\bc} A \to B, \Delta$ by~\ref{eq:supp-imp}.

        \item[$\Delta = \bot, \Delta':$] assume that $\Vdash_{\bb} \bot, \Delta'$. Then $\Vdash_{\bb} \Delta'$ by~\ref{eq:supp-bot}. Then $\Vdash_{\bc} \Delta'$ by induction hypothesis, hence $\Vdash_{\bc} \bot, \Delta'$ by~\ref{eq:supp-bot}. \qed
        
    \end{description}
\end{proof}

This allows to prove the following result about validity:
\begin{theorem}[Validity] \label{thm:validity}
    $\Gamma \Vdash \Delta$ if and only if $\Gamma \Vdash_{\mathcal{ST}} \Delta$.
\end{theorem}
\begin{proof}
    ($\Rightarrow$) Immediate by Definition~\ref{def:validity}---since $\Gamma \Vdash_{\bb} \Delta$ holds for all $\bb \supseteq \st$, it holds for $\mathcal{ST}$ in particular.

    ($\Leftarrow$) Let $\Gamma = \{ A^{1}, \dots, A^{n}\}$ and $\Theta^{1}, \dots, \Theta^{n}$ be arbitrary contexts. Assume that $\Vdash_{\bb} \Theta^{1}, A^{1}, \dots, \Vdash_{\bb} \Theta^{n}, A^{n}$ for arbitrary $\bb \supseteq \st$. Then $\Vdash_{\bb} \Theta^{1}, \dots, \Theta^{n}, \Delta$ by~\ref{eq:supp-inf}. Further, by Lemma~\ref{lm:monotonicity}, $\Vdash_{\bc} \Theta^{1}, \dots, \Theta^{n}, \Delta$, as well as $\Vdash_{\bc} \Theta^{1}, A^{1}, \dots, \Vdash_{\bc} \Theta^{n}, A^{n}$, for arbitrary $\bc \supseteq \bb$. Then $\Gamma \Vdash_{\bb} \Delta$ by~\ref{eq:supp-inf} again, and since $\bb$ is arbitrary, we have that, by Definition~\ref{def:validity},  $\Gamma \Vdash \Delta$.\qed
\end{proof}

We will next demonstrate that, even though~\ref{eq:supp-inf} is defined via atomic sets, its action can be extended to arbitrary contexts.

\begin{lemma} \label{lemma:arbitrarycontexts}
    $\{A^{1},  \ldots , A^{n} \} \Vdash_{\bb} \Delta$ if and only if, for all
    $\bc \supseteq \bb$ and arbitrary sets of formulas $\Theta^{i}$, if for $1 \leq i \leq n$, $\Vdash_{\bc} \Theta^{i},A^{i}$, then $\Vdash_{\bc} \Theta^{1}, \ldots , \Theta^{n}, \Delta$.
\end{lemma}

\begin{proof}
We need to prove that $\{A^{1},  \ldots , A^{n} \} \Vdash_{\bb} \Delta$ if and only if, for all $\bc \supseteq \bb$ and arbitrary sets of formulas $\Theta^{i}$, if $\Vdash_{\bc} \Theta^{i},A^{i}$  for all $1\leq i\leq n$,  then $\Vdash_{\bc} \Theta^{1}, \ldots , \Theta^{n}, \Delta$. \\
($\Leftarrow$): Immediate by~\ref{eq:supp-inf}---since the statement holds for arbitrary sets of formulas $\Theta^{i}$, it in particular holds for all atomic sets $\Theta^{i}_{\At}$. \\
($\Rightarrow$): We prove the result by induction on the degree of $\left\{\Theta^{i}\right\}_{1\leq i\leq n}$, understood as the sum of the degrees of all formulas occurring in all sets in the collection.

Base case: The degree of each $\Theta^{i}$ is $0$, so all sets only contain atoms. Then the result follows immediately by~\ref{eq:supp-inf}, as it is defined via atomic sets.

If the degree of $\Theta^{1}, \ldots , \Theta^{n}$ is not $0$, then there exists some $k$ ($1 \leq k \leq n$) such that $\Theta^{k}$ contains some complex formula $F$. Without loss of generality, let $\Theta^{k} = \{ \Omega, F \}$. Therefore, it is left to be shown that, for an arbitrary $\bc \supseteq \bb$, if $\Vdash_{\bc} \Theta^{1},A^{1}, \dots, \Vdash_{\bc} \Omega, F, A^{k}, \dots, \Vdash_{\bc} \Theta^{n},A^{n}$ then $\Vdash_{\bc} \Theta^{1}, \ldots , \Omega, F, \ldots, \Theta^{n}, \Delta$; as the induction hypothesis, take that this statement holds for any extension of $\bb$ and any strict subformula of $F$ in place of $F$. We prove the result inductively on $F$.

\begin{description}[itemsep=0.2em]
        \item[$F = A \land B:$] we have that $\Vdash_{\bc} \Omega, A \land B, A^{k}$. We obtain $\Vdash_{\bc} \Omega, A, A^{k}$ and $\Vdash_{\bc} \Omega, B, A^{k}$ by~\ref{eq:supp-and}. The induction hypothesis thus yields $\Vdash_{\bc} \Theta^{1}, \ldots , \Omega, A, \ldots, \Theta^{n}, \Delta$ and $\Vdash_{\bc} \Theta^{1}, \ldots , \Omega, B, \ldots, \Theta^{n}, \Delta$. By~\ref{eq:supp-and}, we obtain $\Vdash_{\bc} \Theta^{1}, \ldots , \Omega, A \land B, \ldots, \Theta^{n}, \Delta$ as required.

        \item[$F = A \lor B:$] we have that $\Vdash_{\bc} \Omega, A \lor B, A^{k}$. By~\ref{eq:supp-or} we have $\Vdash_{\bc} \Omega, A, B, A^{k}$. The induction hypothesis yields $\Vdash_{\bc} \Theta^{1}, \ldots , \Omega, A, B, \ldots, \Theta^{n}, \Delta$, so by~\ref{eq:supp-or} we obtain $\Vdash_{\bc} \Theta^{1}, \ldots , \Omega, A \lor B, \ldots, \Theta^{n}, \Delta$ as required.

        \item[$F = A \to B:$] we have that $\Vdash_{\bc} \Omega, A \to B, A^{k}$. By~\ref{eq:supp-imp} we obtain $A \Vdash_{\bc} \Omega, B, A^{k}$. Let $\Sigma_{\At}$ be an arbitrary atomic set and pick an arbitrary $\bd \supseteq \bc$ such that $\Vdash_{\bd} \Sigma_{\At}, A$. Since also $A \Vdash_{\bc} \Omega, B, A^{k}$ and $\bd \supseteq \bc$, we obtain $\Vdash_{\bd} \Sigma_{\At}, \Omega, B, A^{k}$ by~\ref{eq:supp-inf}. By monotonicity we have that $\Vdash_{\bd} \Theta^{j}, A^{j}$ (for $1 \leq j \leq n, j \neq k$), which together with $\Vdash_{\bd} \Sigma_{\At}, \Omega, B, A^{k}$ yield $\Vdash_{\bd} \Theta^{1}, \ldots, \Sigma_{\At}, \Omega, B, \ldots, \Theta^{n}, \Delta$ by the induction hypothesis (as $\bd \supseteq \bb$). Since also $\bd \supseteq \bc$ such that $\Vdash_{\bd} \Sigma_{\At}, A$ for arbitrary $\Sigma_{\At}$, we obtain $A \Vdash_{\bc} \Theta^{1}, \ldots , \Omega, B, \ldots, \Theta^{n}, \Delta$ by~\ref{eq:supp-inf}. By~\ref{eq:supp-imp}, then, $\Vdash_{\bc} \Theta^{1}, \ldots , \Omega, A \to B, \ldots, \Theta^{n}, \Delta$ as required.

        \item[$F = \bot:$] we have that $\Vdash_{\bc} \Omega, \bot, A^{k}$. By~\ref{eq:supp-bot} we obtain $\Vdash_{\bc} \Omega, A^{k}$. The induction hypothesis yields $\Vdash_{\bc} \Theta^{1}, \ldots , \Omega, \ldots, \Theta^{n}, \Delta$, so we get $\Vdash_{\bc} \Theta^{1}, \ldots , \Omega, \bot, \ldots, \Theta^{n}, \Delta$ by~\ref{eq:supp-bot} as required.\qed
        
 \end{description}       
\end{proof}

Finally, we present the admissibility of structural rules.
\begin{lemma}[Right/left weakening] \label{lemma:weak}
    If $\Vdash_{\bb} \Gamma$ then $\Vdash_{\bb} \Gamma, A$. If $\Gamma \Vdash_{\bb} \Delta$ then $A, \Gamma \Vdash_{\bb} \Delta$.
\end{lemma}
\begin{proof}[of Lemma~\ref{lemma:weak}]
Right weakening: We will prove that  if $\Vdash_{\bb} \Gamma$ then $\Vdash_{\bb} \Gamma, A$.
     
    For the base case, consider both $\Gamma$ and $A$ to be atomic, {\ie} the following clause with $\Gamma = \Gamma_{\At}$, $A = p$ for an arbitrary atom $p$: $$\text{if } \Vdash_{\bb} \Gamma_{\At} \text{ then } \Vdash_{\bb} \Gamma_{\At}, p.$$
    By~\ref{eq:supp-at}, we obtain $\vdash_{\bb} \seq\Gamma_{\At}$ from $\Vdash_{\bb} \Gamma_{\At}$. Since the definition of derivability (Definition~\ref{def:derivability}) allows inclusion of arbitrary contexts on the right, we conclude $\vdash_{\bb} \seq\Gamma_{\At}, p$, and, by~\ref{eq:supp-at} again, we obtain $\Vdash_{\bb} \Gamma_{\At}, p$ as desired.

    For the inductive case, notice that it does not make a difference whether to consider $\Gamma$ or $A$ to be non-atomic, as in either case, one would end up with \textit{some} non-atomic formula on the right, be it a member of $\Gamma$ or $A$ itself. What truly matters is the ability to break such a formula down into its strict subformulas until, eventually, one ends up solely with atoms on the right.

    As an induction hypothesis, thus, take that $$\text{if } \Vdash_{\bb} \Gamma \text{ then } \Vdash_{\bb} \Gamma, A$$ holds true with any base in place of $\bb$ and any proper subset of $\Gamma$ in place of $\Gamma$.

    \begin{description}[itemsep=0.2em]
        \item[$\Gamma = \Gamma', B \land C:$] we have that $\Vdash_{\bb} \Gamma', B \land C$. By~\ref{eq:supp-and}, obtain $\Vdash_{\bb} \Gamma', B$ and $\Vdash_{\bb} \Gamma', C$. By the induction hypothesis, $\Vdash_{\bb} \Gamma', B, A$ and $\Vdash_{\bb} \Gamma', C, A$, and by~\ref{eq:supp-and} again, $\Vdash_{\bb} \Gamma', B \land C, A$.

        \item[$\Gamma = \Gamma', B \lor C:$] we have that $\Vdash_{\bb} \Gamma', B \lor C$. By~\ref{eq:supp-or}, obtain $\Vdash_{\bb} \Gamma', B, C$. By the induction hypothesis, $\Vdash_{\bb} \Gamma', B, C, A$, and by~\ref{eq:supp-or} again, $\Vdash_{\bb} \Gamma', B \lor C, A$. 
    
        \item[$\Gamma = \Gamma', B \to C:$] we have that $\Vdash_{\bb} \Gamma', B \to C$. By~\ref{eq:supp-imp}, obtain $B \Vdash_{\bb} \Gamma', C$. Let $\Theta_{\At}$ be an arbitrary atomic set and pick an arbitrary $\bc \supseteq \bb$ such that $\Vdash_{\bc} B, \Theta_{\At}$. Since also $B \Vdash_{\bb} \Gamma', C$, we obtain $\Vdash_{\bc} \Theta_{\At}, \Gamma', C$ by~\ref{eq:supp-inf}. By the induction hypothesis, $\Vdash_{\bc} \Theta_{\At}, \Gamma', C, A$. Since also $\bc \supseteq \bb$ such that $\Vdash_{\bc} B, \Theta_{\At}$ for arbitrary $\Theta_{\At}$, we obtain $B \Vdash_{\bb} \Gamma', C, A$ by~\ref{eq:supp-inf}. By~\ref{eq:supp-imp}, $\Vdash_{\bb} \Gamma', B \to C, A$.

        \item[$\Gamma = \Gamma', \bot:$] we have that $\Vdash_{\bb} \Gamma', \bot$. By~\ref{eq:supp-bot}, obtain $\Vdash_{\bb} \Gamma'$. By the induction hypothesis, $\Vdash_{\bb} \Gamma', A$, and by~\ref{eq:supp-bot} again, $\Vdash_{\bb} \Gamma', \bot, A$. 
    \end{description}

For the left weakening, we need to prove that if $\Gamma \Vdash_{\bb} \Delta$ then $A, \Gamma \Vdash_{\bb} \Delta$. Without loss of generality, let $\Gamma = \{ A^{1}, \dots, A^{n}\}$ and $\Theta^{1}_{\At}, \dots, \Theta^{n}_{\At}, \Omega_{\At}$ be arbitrary atomic sets. Pick an arbitrary $\bc \supseteq \bb$ such that ${\Vdash_{\bc} \Theta^{1}_{\At}, A^{1}, \dots, \Vdash_{\bc} \Theta^{n}_{\At}, A^{n}}$ and ${\Vdash_{\bc} \Omega_{\At}, A}$. By~\ref{eq:supp-inf}, from $\Gamma \Vdash_{\bb} \Delta$ and ${\Vdash_{\bc} \Theta^{1}_{\At}, A^{1}, \dots, \Vdash_{\bc} \Theta^{n}_{\At}, A^{n}}$ we obtain ${\Vdash_{\bc} \Theta^{1}_{\At}, \dots, \Theta^{n}_{\At}, \Delta}$. Application of weakening with $\Omega$ on the right yields $\Vdash_{\bc} \Theta^{1}_{\At}, \dots, \Theta^{n}_{\At}, \Delta, \Omega_{\At}$. Since also $\bc \supseteq \bb$ such that $\Vdash_{\bc} \Theta^{1}_{\At}, A^{1}, \dots, \Vdash_{\bc} \Theta^{n}_{\At}, A^{n}$ and $\Vdash_{\bc} \Omega_{\At}, A$ for arbitrary $\Theta^{1}_{\At}, \dots, \Theta^{n}_{\At}, \Omega_{\At}$, we obtain $A, \Gamma \Vdash_{\bb} \Delta$ by~\ref{eq:supp-inf} as required. \qed
\end{proof}

\begin{lemma}[Right/left contraction]\label{lemma:cont}
    If $\Vdash_{\bb} \Gamma, A, A$ then $\Vdash_{\bb} \Gamma, A$.  If \\$\Gamma, A, A \Vdash_{\bb} \Delta$ then $\Gamma, A \Vdash_{\bb} \Delta$. 
\end{lemma}
\begin{proof}
Right contraction: If $\Vdash_{S} \Gamma, A, A$ then $\Vdash_{S} \Gamma, A$ follows directly from the fact that $\Gamma, A, A$ and $\Gamma, A$ are the same set of formulas.

For the left contraction, we need to prove that  if $\Gamma, A, A \Vdash_{\bb} \Delta$ then $\Gamma, A \Vdash_{\bb} \Delta$. Without loss of generality, let $\Gamma = \{ A^{1}, \dots, A^{n}\}$ and $\Theta^{1}_{\At}, \dots, \Theta^{n}_{\At}, \Omega_{\At}$ be arbitrary atomic sets. Choose an arbitrary $\bc \supseteq \bb$ such that $\Vdash_{\bc} \Theta^{1}_{\At}, A^{1}, \dots, \Vdash_{\bc} \Theta^{n}_{\At}, A^{n}$ and $\Vdash_{\bc} \Omega_{\At}, A$. From $\Gamma, A, A \Vdash_{\bb} \Delta$ and $\Vdash_{\bc} \Theta^{1}_{\At}, A^{1}, \dots, \Vdash_{\bc} \Theta^{n}_{\At}, A^{n}$ and $\Vdash_{\bc} \Omega_{\At}, A$ applied twice we obtain $\Vdash_{\bc} \Theta^{1}_{\At}, \dots, \Theta^{n}_{\At}, \Omega_{\At}, \Omega_{\At}, \Delta$ by~\ref{eq:supp-inf}. Application of contraction on $\Omega_{\At}$ on the right yields $\Vdash_{\bc} \Theta^{1}_{\At}, \dots, \Theta^{n}_{\At}, \Omega_{\At}, \Delta$. Since also $\bc \supseteq \bb$ such that $\Vdash_{\bc} \Theta^{1}_{\At}, A^{1}, \dots, \Vdash_{\bc} \Theta^{n}_{\At}, A^{n}$ and $\Vdash_{\bc} \Omega_{\At}, A$ for arbitrary $\Theta^{1}_{\At}, \dots, \Theta^{n}_{\At}, \Omega_{\At}$, we obtain $\Gamma, A \Vdash_{\bb} \Delta$ by~\ref{eq:supp-inf} as required. \qed
\end{proof}

\section{Soundness}
We shall consider the sequent system $\CLp$  for classical propositional logic with rules depicted in Fig.~\ref{fig:CL}.
It should be highlighted that, although the presentation is non-standard---adopting a multiplicative style for inference rules instead of additive rules---the rules regard {\em sets of formulas}.
\begin{figure}[t]
\begin{center} 
\begin{bprooftree}
    \AxiomC{$A,B,\Gamma \Rightarrow \Delta$}
    \RightLabel{\scriptsize{$L \land$}}
    \UnaryInfC{$A \land B, \Gamma \Rightarrow \Delta$}
\end{bprooftree}
\quad
\begin{bprooftree}
    \AxiomC{$\Gamma \Rightarrow  \Delta, A$}
    \AxiomC{$\Gamma' \Rightarrow \Delta', B$}
    \RightLabel{\scriptsize{$R \land$}}
    \BinaryInfC{$\Gamma, \Gamma' \Rightarrow \Delta, \Delta', A \land B$}
\end{bprooftree}

\end{center}
   
\begin{center}

\begin{bprooftree}
    \AxiomC{$A, \Gamma \Rightarrow  \Delta$}
    \AxiomC{$B, \Gamma' \Rightarrow  \Delta'$}
    \RightLabel{\scriptsize{$L \lor$}}
    \BinaryInfC{$A \lor B, \Gamma, \Gamma' \Rightarrow \Delta, \Delta'$}
\end{bprooftree}
\quad
\begin{bprooftree}
    \AxiomC{$\Gamma \Rightarrow \Delta, A, B$}
    \RightLabel{\scriptsize{$R \lor$}}
    \UnaryInfC{$\Gamma \Rightarrow \Delta, A \lor B$}
\end{bprooftree}
\end{center}

\begin{center}
\begin{bprooftree}
    \AxiomC{$\Gamma \Rightarrow  \Delta, A$}
    \AxiomC{$B, \Gamma' \Rightarrow \Delta'$}
    \RightLabel{\scriptsize{$L \to$}}
    \BinaryInfC{$A \to B, \Gamma, \Gamma' \Rightarrow \Delta, \Delta'$}
\end{bprooftree}
\quad
\begin{bprooftree}
    \AxiomC{$A, \Gamma \Rightarrow \Delta, B$}
    \RightLabel{\scriptsize{$R \to$}}
    \UnaryInfC{$\Gamma \Rightarrow \Delta, A \to B$}
\end{bprooftree}
\end{center}
\begin{center}
\begin{bprooftree}
    \AxiomC{}
    \RightLabel{$\init$}
    \UnaryInfC{$\Gamma, A \Rightarrow A, \Delta$}
\end{bprooftree}
\quad
\begin{bprooftree}
    \AxiomC{}
    \RightLabel{\scriptsize{$L \bot$}}
    \UnaryInfC{$\Gamma, \bot \Rightarrow \Delta$ }
\end{bprooftree}
\quad
\begin{bprooftree}
    \AxiomC{$\Gamma \Rightarrow \Delta$}
    \RightLabel{\scriptsize{$R \bot$}}
    \UnaryInfC{$\Gamma \Rightarrow \bot,\Delta$ }
\end{bprooftree}
\caption{Sequent system $\CLp$.}\label{fig:CL}
\end{center}
\end{figure}

This section is devoted to proving that $\CLp$ is {\em sound} with respect to our semantics; that is, every provable sequent is semantically valid. 
This result follows from the fact that the semantic relation $\Vdash$ preserves all inference rules of $\CLp$.

\begin{theorem}[Soundness] \label{thm:soundness}
    If there is a proof of $\Gamma \seq \Delta$ in $\CLp$ then $\Gamma \Vdash \Delta$.
\end{theorem}
\begin{proof}[of Theorem~\ref{thm:soundness}]
In what follows, we assume without loss of generality that $\Gamma = { A^{1}, \dots, A^{n} }$ and $\Gamma' = { B^{1}, \dots, B^{n} }$. We use $\Theta^{1}_{\At}, \dots, \Theta^{n}_{\At},, \Sigma^{1}_{\At}, \dots, \Sigma^{n}_{\At}, \Omega_{\At}$ to denote arbitrary atomic contexts. We also make use of Theorem~\ref{thm:validity} throughout the proof, {\ie} that $\Gamma \Vdash \Delta$ and $\Gamma \Vdash_{\st} \Delta$ are equivalent.

\begin{description}[itemsep=0.5em, font=\normalfont]

\item[(Ax)'.] Take any $\bb$ and an arbitrary $\bc \supseteq \bb$ such that $\Vdash_{\bc} A, \Omega_{\At}$ and $\Vdash_{\bc} A^{i}, \Theta^{i}_{\At}$ for all $A^{i} \in \Gamma$. Since $\Vdash_{\bc} A, \Omega_{\At}$, by applying right weakening (Lemma \ref{lemma:cont}) a sufficient amount of times we obtain $\Vdash_{\bc} A, \Omega_{\At}, \Theta^{1}_{\At}, \ldots, \Theta^{n}_{\At}, \Delta$. Since also $\bc \supseteq \bb$ such that $\Vdash_{\bc} A, \Omega_{\At}$ and $\Vdash_{\bc} A^{i}, \Theta^{i}_{\At}$ for arbitrary $\Theta^{i}_{\At}, \Omega_{\At}$, we conclude $\Gamma, A \Vdash_{\bb} A, \Delta$ by~\ref{eq:supp-inf}, so by arbitrariness of $\bb$ also $\Gamma, A \Vdash A, \Delta$.

\item[($L \land$)'.] Assume $A,B,\Gamma \Vdash \Delta$ and, for an arbitrary $\bb \supseteq \st$, that $\Vdash_{\bb} \Theta^{i}_{\At}, A^{i}$ for all $A^{i} \in \Gamma$ and $\Vdash_{\bb} \Omega_{\At}, A \land B$. Then, by~\ref{eq:supp-and} we obtain $\Vdash_{\bb} \Omega_{\At}, A$ and $\Vdash_{\bb} \Omega_{\At}, B$. Now, $A,B,\Gamma \Vdash_{\mathcal{ST}} \Delta$ together with $\Vdash_{\bb} \Theta^{1}_{\At}, A^{1}, \dots, \Vdash_{\bb} \Theta^{n}_{\At}, A^{n}$ and $\Vdash_{\bb} \Omega_{\At}, A$ and $\Vdash_{\bb} \Omega_{\At}, B$ gives us $\Vdash_{\bb} \Theta^{1}_{\At}, \dots, \Theta^{n}_{\At}, \Omega_{\At}, \Delta$ by~\ref{eq:supp-inf} and right contraction on $\Omega_{\At}$ (Lemma~\ref{lemma:cont}). Since also $\bb \supseteq \st$ such that $\Vdash_{\bb} \Theta^{i}_{\At}, A^{i}$ for all $A^{i} \in \Gamma$ and $\Vdash_{\bb} \Omega_{\At}, A \land B$ for arbitrary $\Theta^{i}_{\At}, \Omega_{\At}$, we conclude $A \land B, \Gamma \Vdash_{\mathcal{ST}} \Delta$ by~\ref{eq:supp-inf}, hence $A \land B, \Gamma \Vdash \Delta$ as required. 

\item[($R \land$)'.] Assume that $\Gamma \Vdash  \Delta, A$ and $\Gamma' \Vdash \Delta', B$ and, for an arbitrary $\bb \supseteq \st$, that $\Vdash_{\bb} \Theta^{i}_{\At}, A^{i}$ for all $A^{i} \in \Gamma$ and $\Vdash_{\bb} \Sigma^{i}_{\At}, B^{i}$ for all $B^{i} \in \Gamma'$. We then, respectively, conclude $\Vdash_{\bb} \Theta^{1}_{\At}, \dots, \Theta^{n}_{\At}, \Delta, A$ and $\Vdash_{\bb} \Sigma^{1}_{\At}, \dots, \Sigma^{n}_{\At}, \Delta', B$ by~\ref{eq:supp-inf}. By right weakening (Lemma~\ref{lemma:weak}), we obtain $\Vdash_{\bb} \Theta^{1}_{\At}, \dots, \Theta^{n}_{\At}, \Sigma^{1}_{\At}, \dots, \Sigma^{n}_{\At}, \Delta, \Delta', A$ and $\Vdash_{\bb} \Theta^{1}_{\At}, \dots, \Theta^{n}_{\At}, \Sigma^{1}_{\At}, \dots, \Sigma^{n}_{\At}, \Delta, \Delta', B$, from which, by~\ref{eq:supp-and}, we further conclude $\Vdash_{\bb} \Theta^{1}_{\At}, \dots, \Theta^{n}_{\At}, \Sigma^{1}_{\At}, \dots, \Sigma^{n}_{\At}, \Delta, \Delta', A \land B$. Since also $\bb \supseteq \st$ such that $\Vdash_{\bb} \Theta^{i}_{\At}, A^{i}$ for all $A^{i} \in \Gamma$ and $\Vdash_{\bb} \Sigma^{i}_{\At}, B^{i}$ for all $B^{i} \in \Gamma'$ for arbitrary $\Theta^{i}_{\At}, \Sigma^{i}_{\At}$, we obtain $\Gamma, \Gamma' \Vdash_{\mathcal{ST}} \Delta, \Delta', A \land B$ by~\ref{eq:supp-inf}, hence $\Gamma, \Gamma' \Vdash \Delta, \Delta', A \land B$ as required. 

\item[($L \lor$)'.] Assume that $\Gamma, A \Vdash \Delta$ and $\Gamma', B \Vdash \Delta'$ and, for an arbitrary $\bb \supseteq \st$, that $\Vdash_{\bb} \Theta^{i}_{\At}, A^{i}$ for all $A^{i} \in \Gamma$, $\Vdash_{\bb} \Sigma^{i}_{\At}, B^{i}$ for all $B^{i} \in \Gamma'$ and $\Vdash_{\bb} \Omega_{\At}, {A \lor B}$. We conclude $\Vdash_{\bb} \Omega_{\At}, A, B$ by~\ref{eq:supp-or}. Now, $\Gamma', B \Vdash_{\mathcal{ST}} \Delta'$ and $\Vdash_{\bb} \Omega_{\At}, A, B$ together with $\Vdash_{\bb} \Sigma^{1}_{\At}, B^{1}, \dots, \Vdash_{\bb} \Sigma^{n}_{\At}, B^{n}$ give us $\Vdash_{\bb} \Sigma^{1}_{\At}, \dots,\Sigma^{n}_{\At}, \Omega_{\At}, A, \Delta'$ by Lemma~\ref{lemma:arbitrarycontexts}. Since also $\Gamma, A \Vdash_{\mathcal{ST}}\Delta$ and $\Vdash_{\bb} \Theta^{1}_{\At}, A^{1}, \dots, \Vdash_{\bb} \Theta^{n}_{\At}, A^{n}$, we obtain $\Vdash_{\bb} \Theta^{1}_{\At}, \dots, \Theta^{n}_{\At}, \Sigma^{1}_{\At}, \dots, \Sigma^{n}_{\At}, \Omega_{\At}, \Delta', \Delta$ by Lemma~\ref{lemma:arbitrarycontexts}. On top of that, since $\bb \supseteq \st$ such that $\Vdash_{\bb} \Theta^{i}_{\At}, A^{i}$ for all $A^{i} \in \Gamma$, $\Vdash_{\bb} \Sigma^{i}_{\At}, B^{i}$ for all $B^{i} \in \Gamma'$ and $\Vdash_{\bb} \Omega_{\At}, A \lor B$ for arbitrary $\Theta^{i}_{\At}, \Sigma^{i}_{\At}, \Omega_{\At}$, we obtain $\Gamma, \Gamma', A \lor B \Vdash_{\mathcal{ST}} \Delta, \Delta'$ by~\ref{eq:supp-inf}, hence $\Gamma, \Gamma', A \lor B \Vdash \Delta, \Delta'$ as required. 

\item[($R \lor$)'.] Assume that $\Gamma \Vdash \Delta, A, B$ and, for an arbitrary $\bb \supseteq \st$, that $\Vdash_{\bb} \Theta^{i}_{\At}, A^{i}$ for all $A^{i} \in \Gamma$. We then get $\Vdash_{\bb} \Theta^{1}_{\At}, \dots, \Theta^{n}_{\At}, \Delta, A, B$ by~\ref{eq:supp-inf}. Then, by~\ref{eq:supp-or} we obtain $\Vdash_{\bb} \Theta^{1}_{\At}, \dots, \Theta^{n}_{\At}, \Delta, A \lor B$. 
Since also $\bb \supseteq \st$ such that $\Vdash_{\bb} \Theta^{i}_{\At}, A^{i}$ for all $A^{i} \in \Gamma$ for arbitrary $\Theta^{i}_{\At}$, we conclude $\Gamma \Vdash_{\mathcal{ST}} \Delta, A \lor B$ by~\ref{eq:supp-inf}, hence $\Gamma \Vdash \Delta, A \lor B$ as required. 

\item[($L \to$)'.] Assume that $\Gamma \Vdash  \Delta, A$ and $B, \Gamma' \Vdash \Delta'$ and, for an arbitrary ${\bb \supseteq \st}$, that $\Vdash_{\bb} \Theta^{i}_{\At}, A^{i}$ for all $A^{i} \in \Gamma$, $\Vdash_{\bb} \Sigma^{i}_{\At}, B^{i}$ for all $B^{i} \in \Gamma'$ and $\Vdash_{\bb} \Omega_{\At}, A \to B$. Then, by~\ref{eq:supp-imp}, we have $A \Vdash_{\bb} \Omega_{\At}, B$. Now, by~\ref{eq:supp-inf}, from $\Gamma \Vdash_{\mathcal{ST}} \Delta, A$ and $\Vdash_{\bb} \Theta^{1}_{\At}, A^{1}, \dots, \Vdash_{\bb} \Theta^{n}_{\At}, A^{n}$ we obtain $\Vdash_{\bb} \Theta^{1}_{\At}, \dots, \Theta^{n}_{\At}, \Delta, A$. Since also $A \Vdash_{\bb} \Omega_{\At}, B$, we conclude $\Vdash_{\bb} \Theta^{1}_{\At}, \dots, \Theta^{n}_{\At}, \Delta, \Omega_{\At}, B$ by Lemma~\ref{lemma:arbitrarycontexts}. Since we also have $B, \Gamma' \Vdash_{\mathcal{ST}} \Delta'$ and $\Vdash_{\bb} \Sigma^{1}_{\At}, B^{1}, \dots, \Vdash_{\bb} \Sigma^{n}_{\At}, B^{n}$, we conclude ${\Vdash_{\bb} \Theta^{1}_{\At}, \dots, \Theta^{n}_{\At}, \Sigma^{1}_{\At}, \dots, \Sigma^{n}_{\At}, \Omega_{\At}, \Delta, \Delta'}$ by Lemma~\ref{lemma:arbitrarycontexts} again. On top of that, since ${\bb \supseteq \st}$ such that $\Vdash_{\bb} \Theta^{i}_{\At}, A^{i}$ for all $A^{i} \in \Gamma$, $\Vdash_{\bb} \Sigma^{i}_{\At}, B^{i}$ for all $B^{i} \in \Gamma'$ and $\Vdash_{\bb} \Omega_{\At}, A \to B$ for arbitrary $\Theta^{i}_{\At}, \Sigma^{i}_{\At}, \Omega_{\At}$, we get ${\Gamma, \Gamma', A \to B \Vdash_{\mathcal{ST}} \Delta, \Delta'}$ by~\ref{eq:supp-inf}, hence $\Gamma, \Gamma', A \to B \Vdash \Delta, \Delta'$ as required. 

\item[($R \to$)'.] Assume that $A, \Gamma \Vdash B, \Delta$ and, for an arbitrary $\bb \supseteq \st$, that ${\Vdash_{\bb} \Theta^{i}_{\At}, A^{i}}$ for all $A^{i} \in \Gamma$. Further, pick an arbitrary $\bc \supseteq \bb$ such that $\Vdash_{\bc} \Omega_{\At}, A$. We then get $\Vdash_{\bc} \Theta^{1}_{\At}, \dots, \Theta^{n}_{\At}, \Omega_{\At}, B, \Delta$ by~\ref{eq:supp-inf} and monotonicity (Lemma~\ref{lm:monotonicity}). Since also $\bc \supseteq \bb$ such that $\Vdash_{\bc} \Omega_{\At}, A$, we obtain $A \Vdash_{\mathcal{B}} \Theta^{1}_{\At}, \dots, \Theta^{n}_{\At}, B, \Delta$ by~\ref{eq:supp-inf}. By~\ref{eq:supp-imp} we then obtain ${\Vdash_{\bb} \Theta^{1}_{\At}, \dots, \Theta^{n}_{\At}, A \to B, \Delta}$. Since also ${\bb \supseteq \st}$ such that $\Vdash_{\bb} \Theta^{i}_{\At}, A^{i}$ for all $A^{i} \in \Gamma$ for arbitrary $\Theta^{i}_{\At}$, we obtain $\Gamma \Vdash_{\mathcal{ST}} A \to B, \Delta$ by~\ref{eq:supp-inf}, hence $\Gamma \Vdash A \to B, \Delta$ as required. 

\item[($L \bot$)'.] Take any $\bb$ and any $\bc \supseteq \bb$ such that $\Vdash_{\bc} \bot, \Omega_{\At}$ and $\Vdash_{\bc} A^{i}, \Theta^{i}_{\At}$ for all ${A^{i} \in \Gamma}$. Since $\Vdash_{\bc} \bot, \Omega_{\At}$, by~\ref{eq:supp-bot} we get $\Vdash_{\bc} \Omega_{\At}$, so by applying right weakening (Lemma~\ref{lemma:weak}) a sufficient amount of times we obtain $\Vdash_{\bc} \Omega_{\At}, \Theta^{1}_{\At}, \ldots, \Theta^{n}_{\At}, \Delta$. Since also $\bc \supseteq \bb$ such that $\Vdash_{\bc} A^{i}, \Theta^{i}_{\At}$ for all $A^{i} \in \Gamma$ and $\Vdash_{\bc} \bot, \Omega_{\At}$ for arbitrary $\Theta^{i}_{\At}, \Omega_{\At}$, we conclude $\Gamma, \bot\Vdash_{\bb} \Delta$ by~\ref{eq:supp-inf}, so by arbitrariness of $\bb$ also $\Gamma, \bot \Vdash \Delta$. 

\item[($R \bot$)'.] Assume that $\Gamma \Vdash \Delta$ and, for an arbitrary $\bb \supseteq \st$, that ${\Vdash_{\bb} \Theta^{i}_{\At}, A^{i}}$ for all $A^{i} \in \Gamma$. We then get $\Vdash_{\bb} \Theta^{1}_{\At}, \dots, \Theta^{n}_{\At}, \Delta$ by~\ref{eq:supp-inf}. By~\ref{eq:supp-bot} we then conclude ${\Vdash_{\bb} \Theta^{1}_{\At}, \dots, \Theta^{n}_{\At}, \Delta, \bot}$. Further, since ${\bb \supseteq \st}$ such that $\Vdash_{\bb} \Theta^{i}_{\At}, A^{i}$ for all $A^{i} \in \Gamma$ for arbitrary $\Theta^{i}_{\At}$, we obtain $\Gamma \Vdash_{\mathcal{ST}} \Delta, \bot$ by~\ref{eq:supp-inf}, hence $\Gamma \Vdash \Delta, \bot$ as required. \qed 
\end{description}
\end{proof}

\section{Completeness}
In this section, we prove that $\CLp$ is complete with respect to the semantics; that is, every semantically valid sequent is provable. We begin be returning to Definition \ref{def:derivability} and proving a key lemma about composing derivable sequents.

\begin{lemma}\label{lemma:categoricalreducestohypothetical}
  Let $\bb$ be a base. If, for all $\bc \supseteq \bb$, $\vdash_{\bc} \ \Rightarrow \Gamma_{\At}$ implies $\vdash_{\bc} \ \Rightarrow \Delta_{\At}$, then $\vdash_{\bb} \Gamma_{\At} \Rightarrow \Delta_{\At}$.
\end{lemma}

\begin{proof}
    Assume that for all $\bc \supseteq \bb$, $\vdash_{\bb} \ \Rightarrow \Gamma_{\At}$ implies $\vdash_{\bc} \ \Rightarrow \Delta_{\At}$, and let $\bd$ be the base obtained by adding to $\bb$ an atomic axiom with conclusion $\Rightarrow \Gamma_{\At}$. Then we have $\vdash_{\bd} \ \Rightarrow \Gamma_{\At}$, so from our assumption follows  $\vdash_{\bd} \ \Rightarrow \Delta_{\At}$.  Now consider the deduction $\Pi$ showing  $\vdash_{\bd} \ \Rightarrow \Delta_{\At}$. If $\Pi$ does not contain any applications of the new rule it is also a deduction showing $\vdash_{\bb} \ \Rightarrow \Delta_{\At}$, so also $\vdash_{\bb} \Gamma_{\At} \Rightarrow \Delta_{\At}$ by Definition~\ref{def:derivability}. If $\Pi$ contains applications of the new rule, let $\Pi'$ be the deduction obtained by adding $\Gamma_{\At}$ to the left side of the conclusion of any application of the new rule and any rule application below it in the deduction. To exemplify it graphically, the transformation of $\Pi$ into $\Pi'$ is as follows:

\begin{center}
\begin{bprooftree}
    \AxiomC{}
    \UnaryInfC{$\Rightarrow \Gamma_{\At}$}
    \noLine
    \UnaryInfC{$\Pi$}
    \noLine
    \UnaryInfC{$ \Rightarrow \Delta_{\At}$}
\end{bprooftree}
is transformed into
\begin{bprooftree}
      \AxiomC{}
    \UnaryInfC{$\Gamma_{\At} \Rightarrow \Gamma_{\At}$}
    \noLine
    \UnaryInfC{$\Pi'$}
    \noLine
    \UnaryInfC{$\Gamma_{\At} \Rightarrow \Delta_{\At}$}
\end{bprooftree}
\end{center}
 
\noindent
It is straightforward to check that, since in $\Pi'$ every instance of the new rule is replaced by an atomic axiom of shape $\Gamma_{\At} \Rightarrow \Gamma_{\At}$ already contained in $\bb$ and the new context is also added to all inferences below each new atomic axiom, $\Pi'$ is the desired deduction showing $\vdash_{\bb} \Gamma_{\At} \Rightarrow \Delta_{\At}$.\qed
\end{proof}

Completeness in $\Bes$ is typically shown by associating a unique atom $p^A$ to each subformula $A$ of $\Gamma \cup \Delta$. Exploiting the validity of~$\Gamma \Vdash \Delta$ with respect to every base, we then construct a simulation base $\mathcal{U}$ for $\Gamma \cup \Delta$ in which each~$p^A$ mimics the behaviour of~$A$ in~$\CLp$.

\begin{definition}[Atomic mapping] Let $\Sigma$ be a set of formulas. Let $\Sigma_{S}$ be the set of all subformulas of formulas in $\Sigma$. We say that a function $\alpha: \Sigma_{S} \rightarrow \At$ is an \emph{atomic mapping} for $\Sigma$ if (1) $\alpha$ is injective, (2) $\alpha(p) = p$ for $p \in \At$. If $A\notin\At$, we denote $\alpha(A) \eqcolon p^A$.
\end{definition}

\begin{definition}[Simulation base]
Let $\Sigma$ be a set of formulas, and let $\alpha$ be an atomic mapping on $\Sigma$. A simulation base $\bu$ for $\Sigma$ and $\alpha$ is a base that includes exactly the rules in Fig.~\ref{fig:comp}, instantiated for all $A, B \in \Sigma$ and all multisets $\Gamma_{\At}, \Delta_{\At}, \Gamma_{\At}', \Delta_{\At}'$. We write $\ust$ (resp. $\uhs$) for the closure of $\bu$ under $\st$ (resp. $\hs$).
\end{definition}

\begin{figure}[t]
\[\infer[L \land]{p^{A \land B}, \Gamma_{\At} \Rightarrow \Delta_{\At}}{p^{A},p^{B},\Gamma_{\At} \Rightarrow \Delta_{\At}}
\quad 
\infer[R \land]{\Gamma_{\At}, \Gamma_{\At}' \Rightarrow \Delta_{\At}, \Delta_{\At}', p^{A \land B}}{\Gamma_{\At} \Rightarrow  \Delta_{\At}, p^{A}& \Gamma_{\At}' \Rightarrow \Delta_{\At}', p^{B}}
\]
\\
\[
\infer[L\lor]{p^{A \lor B}, \Gamma_{\At}, \Gamma_{\At}' \Rightarrow \Delta_{\At}, \Delta_{\At}'}{p^{A}, \Gamma_{\At} \Rightarrow  \Delta_{\At} & p^{B}, \Gamma_{\At}' \Rightarrow  \Delta_{\At}'}\quad
\infer[R \lor]{\Gamma_{\At} \Rightarrow \Delta_{\At}, p^{A \lor B}}{\Gamma_{\At} \Rightarrow \Delta_{\At}, p^{A}, p^{B}}
\quad
\infer[R \bot]{\Gamma_{\At} \Rightarrow  \Delta_{\At},p^{\bot}}{\Gamma_{\At} \Rightarrow \Delta_{\At}}
\]
\\
\[
\infer[L \to]{p^{A \to B}, \Gamma_{\At}, \Gamma_{\At}' \Rightarrow \Delta_{\At}, \Delta_{\At}'}{\Gamma_{\At} \Rightarrow  \Delta_{\At}, p^{A} & p^{B}, \Gamma_{\At}' \Rightarrow \Delta_{\At}'}
\quad 
\infer[R \to]{\Gamma_{\At} \Rightarrow \Delta_{\At}, p^{A \to B}}{p^{A}, \Gamma_{\At} \Rightarrow \Delta_{\At}, p^{B}}
\quad
\infer[L \bot]{\Gamma_{\At}, p^{\bot} \Rightarrow \Delta_{\At}}{}
\]
\caption{Rules of the simulation base $\bu$.}\label{fig:comp}
\end{figure}

We are now ready to present how the proof simulation works and prove the completeness result.
\begin{lemma}\label{mainlemma}
    Let $\Theta_{\At}$ be any (possibly empty) set of atoms, $\Sigma$ any (non-empty) set containing only subformulas of $\Gamma \cup \Delta$ and $\Sigma_{\At} = \{p^{B} | B \in \Sigma\}$.  Then, for all $\bb \supseteq \ust$, $\Vdash_{\bb} \Sigma, \Theta_{\At}$ if and only if $\vdash_{\bb} \ \Rightarrow \Sigma_{\At}, \Theta_{\At}$.
\end{lemma}

\begin{proof}
We prove the result by induction on the complexity of $\Sigma \cup \Theta_{\At}$, understood as the number of logical connectives occurring in $\Sigma \cup \Theta_{\At}$.
    \begin{enumerate}
        \item Base case: the complexity of $\Sigma \cup \Theta_{\At}$ is $0$. Then $\Sigma \cup \Theta_{\At}$ is just a set of atoms. Notice that, by virtue of how the mapping is defined, $\alpha(p) = p$ and $p^{p} = p$ for $p \in \At$, so clearly $\Sigma = \Sigma_{\At}$ and the results follow by~\ref{eq:supp-at}.
        \medskip

        If the complexity of $\Sigma \cup \Theta_{\At}$ is greater than $0$, then $\Sigma$ contains at least one formula of shape $A \land B$, $A \lor B$ or $A \to B$. We treat each case separately:

        \bigskip
        
        \item $\Sigma = A \land B \cup \Omega$. 

($\Rightarrow$): Assume $\Vdash_{\bb} A \land B, \Omega, \Theta_{\At}$. By~\ref{eq:supp-and} we have both $\Vdash_{\bb} A, \Omega, \Theta_{\At}$ and $\Vdash_{\bb}  B, \Omega, \Theta_{\At}$. Notice that, since $A \land B \cup \Omega$ is a set of subformulas of $\Gamma \cup \Delta$, both $A \cup \Omega$ and $B \cup \Omega$ are also sets of subformulas of $\Gamma \cup \Delta$. Induction hypothesis: $\vdash_{\bb} \ \Rightarrow p^{A}, \Omega_{\At}, \Theta_{\At}$ and $\vdash_{\bb} \ \Rightarrow p^{B}, \Omega_{\At}, \Theta_{\At}$. Then we can obtain the desired proof of $\vdash_{\bb} \Rightarrow p^{A \land B}, \Omega_{\At}, \Theta_{\At}$ as follows:

\bigskip

\begin{center}
\begin{bprooftree}
    \AxiomC{.}
    \noLine
    \UnaryInfC{.}
        \noLine
    \UnaryInfC{.}
        \noLine
    \UnaryInfC{$\Rightarrow p^{A}, \Omega_{\At}, \Theta_{\At}$}
        \AxiomC{.}
    \noLine
    \UnaryInfC{.}
        \noLine
    \UnaryInfC{.}
        \noLine
    \UnaryInfC{$\Rightarrow p^{B}, \Omega_{\At}, \Theta_{\At}$}
    \RightLabel{\scriptsize{$R \land$}}
    \BinaryInfC{$\Rightarrow p^{A \land B}, \Omega_{\At}, \Theta_{\At}, \Omega_{\At}, \Theta_{\At}$}
\end{bprooftree}
\end{center}

\bigskip

($\Leftarrow$): Assume $\vdash_{\bb} \ \Rightarrow p^{A \land B}, \Omega_{\At}, \Theta_{\At}$. First we do the following:

\bigskip

\begin{center}
\begin{bprooftree}
        \AxiomC{.}
    \noLine
    \UnaryInfC{.}
        \noLine
    \UnaryInfC{.}
        \noLine
    \UnaryInfC{$\Rightarrow p^{A \land B},\Omega_{\At}, \Theta_{\At}$}
\AxiomC{}
\RightLabel{\scriptsize{Axiom}}
\UnaryInfC{$p^{A} \Rightarrow p^{A}$}
\RightLabel{\scriptsize{$L \land$}}
\UnaryInfC{$p^{A \land B} \Rightarrow p^{A}$ }
\RightLabel{\scriptsize{Cut}}
\BinaryInfC{$\Rightarrow p^{A}, \Omega_{\At}, \Theta_{\At}$}
\end{bprooftree}

\bigskip
\bigskip
\bigskip

\begin{bprooftree}
        \AxiomC{.}
    \noLine
    \UnaryInfC{.}
        \noLine
    \UnaryInfC{.}
        \noLine
    \UnaryInfC{$\Rightarrow p^{A \land B}, \Omega_{\At}, \Theta_{\At}$}
\AxiomC{}
\RightLabel{\scriptsize{Axiom}}
\UnaryInfC{$p^{B} \Rightarrow p^{B}$}
\RightLabel{\scriptsize{$L \land$}}
\UnaryInfC{$p^{A \land B} \Rightarrow p^{B}$ }
\RightLabel{\scriptsize{Cut}}
\BinaryInfC{$\Rightarrow p^{B}, \Omega_{\At}, \Theta_{\At}$}
\end{bprooftree}
\end{center}

\bigskip

This shows we have both $\vdash_{\bb} \ \Rightarrow p^{A}, \Omega_{\At}, \Theta_{\At}$ and $\vdash_{\bb} \ \Rightarrow p^{B}, \Omega_{\At}, \Theta_{\At}$. Induction hypothesis: $\Vdash_{\bb} A, \Omega, \Theta_{\At}$ and $\Vdash_{\bb} B, \Omega, \Theta_{\At}$. Then by~\ref{eq:supp-and} we get $\Vdash_{\bb} A \land B, \Omega, \Theta_{\At}$.

\bigskip

\item $\Sigma = A \lor B \cup \Omega$.

\bigskip

($\Rightarrow$): Assume $\Vdash_{\bb} A \lor B, \Omega, \Theta_{\At}$. By~\ref{eq:supp-or} we have $\Vdash_{\bb} A, B, \Omega, \Theta_{\At}$. Notice that, since $A \lor B \cup \Omega$ is a subformula of $\Gamma \cup \Delta$, $A \cup B \cup \Omega$ is also a set of subformulas of $\Gamma \cup \Delta$. Induction hypothesis: $\vdash_{\bb} \ \Rightarrow p^{A}, p^{B}, \Omega_{\At}, \Theta_{\At}$. Then we can obtain a proof of $\vdash_{\bb} \ \Rightarrow p^{A \lor B}, \Omega_{\At}, \Theta_{\At}$ as follows:

\bigskip

\begin{center}
\begin{bprooftree}
    \AxiomC{.}
    \noLine
    \UnaryInfC{.}
        \noLine
    \UnaryInfC{.}
        \noLine
    \UnaryInfC{$\Rightarrow p^{A}, p^{B}, \Omega_{\At}, \Theta_{\At}$}
    \RightLabel{\scriptsize{$R \lor$}}
    \UnaryInfC{$\Rightarrow p^{A \lor B}, \Omega_{\At}, \Theta_{\At}$}
\end{bprooftree}
\end{center}

\bigskip

($\Rightarrow$): Assume $\vdash_{\bb} \ \Rightarrow p^{A \lor B}, \Omega_{\At}, \Theta_{\At}$. We do the following:

\bigskip

\begin{center}
\begin{bprooftree}
        \AxiomC{.}
    \noLine
    \UnaryInfC{.}
        \noLine
    \UnaryInfC{.}
        \noLine
    \UnaryInfC{$\Rightarrow p^{A \lor B}, \Omega_{\At}, \Theta_{\At}$}
\AxiomC{}
\RightLabel{\scriptsize{Axiom}}
\UnaryInfC{$p^{A} \Rightarrow p^{A}$}
\AxiomC{}
\RightLabel{\scriptsize{Axiom}}
\UnaryInfC{$p^{B} \Rightarrow p^{B}$}
\RightLabel{\scriptsize{$L \lor$}}
\BinaryInfC{$ p^{A \lor B} \Rightarrow p^{A}, p^{B}$}
\RightLabel{\scriptsize{Cut}}
\BinaryInfC{$\Rightarrow p^{A}, p^{B}, \Omega_{\At}, \Theta_{\At}$}
\end{bprooftree}
\end{center}

\bigskip

So we conclude $\vdash_{\bb} \ \Rightarrow p^{A}, p^{B}, \Omega_{\At}, \Theta_{\At}$. Then the induction hypothesis yields $\Vdash_{\bb} A, B, \Omega, \Theta_{\At}$, so by~\ref{eq:supp-or}, $\Vdash_{\bb} A \lor B, \Omega, \Theta_{\At}$.

\bigskip

\item $\Sigma = A \to B \cup \Omega$

\bigskip

($\Rightarrow$): Assume $\Vdash_{\bb} A \to B, \Omega, \Theta_{\At}$. Then by~\ref{eq:supp-imp} we have $A \Vdash_{\bb} B, \Omega, \Theta_{\At}$. Now pick any $\bc \supseteq \bb$ with $\vdash_{\bc} \Rightarrow p^{A}$. By induction hypothesis we have $\Vdash_{\bc} A$, so since $A \Vdash_{\bb} B, \Omega, \Theta_{\At}$ and $\bc \supseteq \bb$ we also have $\Vdash_{\bc} B, \Omega, \Theta_{\At}$, so the induction hypothesis yields $\vdash_{\bc} \Rightarrow p^{B}, \Omega_{\At}, \Theta_{\At}$. But then since $\bc$ is an arbitrary extension of $\bb$ with $\vdash_{\bc} \Rightarrow p^{A}$ by Lemma \ref{lemma:categoricalreducestohypothetical} we conclude $\vdash_{\bb} p^{A} \Rightarrow p^{B}, \Omega_{\At}, \Theta_{\At}$, and we can finally obtain the desired proof of $\vdash_{\bb} \Rightarrow p^{A \to B}, \Omega_{\At}, \Theta_{\At}$ as follows:

\bigskip

\begin{center}
\begin{bprooftree}
    \AxiomC{.}
    \noLine
    \UnaryInfC{.}
        \noLine
    \UnaryInfC{.}
        \noLine
    \UnaryInfC{$p^{A} \Rightarrow p^{B}, \Omega_{\At}, \Theta_{\At}$}
    \RightLabel{\scriptsize{$R \to$}}
    \UnaryInfC{$\Rightarrow p^{A \to B}, \Omega_{\At}, \Theta_{\At}$}
\end{bprooftree}
\end{center}

\bigskip

($\Leftarrow$):Assume $\vdash_{\bb} \ \Rightarrow p^{A \to B}, \Omega_{\At}, \Theta_{\At}$, and let $\bc$ be any $\bc\supseteq \bb$ with $\Vdash_{\bc} A, \Pi_{\At}$ for some set of atoms $\Pi_{\At}$. Then by induction hypothesis $\vdash_{\bc} \Rightarrow p^{A}, \Pi_{\At}$ and, since $\bc$ is an extension of $\bb$, we also have $\vdash_{\bc} \ \Rightarrow p^{A \to B},  \Omega_{\At}, \Theta_{\At}$. Then we can obtain the following deduction in $\bc$:

\bigskip

\begin{center}
\begin{bprooftree}
      \AxiomC{.}
    \noLine
    \UnaryInfC{.}
        \noLine
    \UnaryInfC{.}
        \noLine
    \UnaryInfC{$\Rightarrow p^{A}, \Pi_{\At}$}
    \AxiomC{.}
    \noLine
    \UnaryInfC{.}
        \noLine
    \UnaryInfC{.}
        \noLine
    \UnaryInfC{$\Rightarrow p^{A \to B},  \Omega_{\At}, \Theta_{\At}$}
    \AxiomC{}
    \RightLabel{\scriptsize{Axiom}}
    \UnaryInfC{$p^{A} \Rightarrow p^{A}$}
    \AxiomC{}
    \RightLabel{\scriptsize{Axiom}}
    \UnaryInfC{$p^{B} \Rightarrow p^{B}$}
    \BinaryInfC{$p^{A}, p^{A \to B} \Rightarrow p^{B}$}
    \RightLabel{\scriptsize{$L \to$}}
    \BinaryInfC{$p^{A} \Rightarrow p^{B},  \Omega_{\At}, \Theta_{\At}$}
    \RightLabel{\scriptsize{Cut}}
\BinaryInfC{$\Rightarrow p^{B},  \Omega_{\At}, \Theta_{\At}, \Pi_{\At}$}
\end{bprooftree}
\end{center}

\bigskip

So $\vdash_{\bc} \ \Rightarrow p^{B}, \Omega_{\At}, \Theta_{\At}, \Pi_{\At}$, and by induction hypothesis $\Vdash_{\bc} B,\Omega, \Theta_{\At}, \Pi_{\At}$ (as $\Theta_{\At} \cup \Pi_{\At}$ is still a set of atoms). But $\bc$ is an arbitrary extension of $\bb$ with $\Vdash_{\bc} A, \Pi_{\At}$ for an arbitrary set of atoms $\Pi_{\At}$, so we conclude $A \Vdash_{\bb} B, \Omega, \Theta_{\At}$ by~\ref{eq:supp-imp} and then $\Vdash_{\bb} A \to B, \Omega, \Theta_{\At}$.

\bigskip

\item $\Sigma = \bot \cup \Omega$

\bigskip

($\Rightarrow$): Assume $\Vdash_{\bb} \bot, \Omega, \Theta_{\At}$. Then by~\ref{eq:supp-bot} we have $\Vdash_{\bb} \Omega, \Theta_{\At}$. The induction hypothesis yields $\Vdash_{\bb} \Omega_{\At}, \Theta_{\At}$, whence $\vdash_{\bb} \Omega_{\At}, \Theta_{\At}$ so from the definition of derivability (Definition~\ref{def:derivability}) we obtain $\vdash_{\bb} p^{\bot}, \Omega_{\At}, \Theta_{\At}$, thus $\Vdash_{\bb} p^{\bot}, \Omega_{\At}, \Theta_{\At}$.

\bigskip

($\Leftarrow$): Assume $\Vdash_{\bb} p^{\bot}, \Omega_{\At}, \Theta_{\At}$. We put $\Gamma = \Delta = \varnothing$ in $L \bot$ and do the following:

\bigskip

\begin{center}
\begin{bprooftree}
      \AxiomC{.}
    \noLine
    \UnaryInfC{.}
        \noLine
    \UnaryInfC{.}
        \noLine
    \UnaryInfC{$\Rightarrow p^{\bot}, \Omega_{\At}, \Theta_{\At}$}
    \AxiomC{}
    \RightLabel{\scriptsize{$L \bot$}}
    \UnaryInfC{$p^{\bot} \Rightarrow \varnothing$}
        \RightLabel{\scriptsize{Cut}}
    \BinaryInfC{$\Rightarrow \Omega_{\At}, \Theta_{\At}$}
\end{bprooftree}
\end{center}

\bigskip

So $\vdash_{\mathcal{B}} \Omega_{\At}, \Theta_{\At}$. The induction hypothesis yields $\Vdash_{\mathcal{B}} \Omega, \Theta_{\At}$, thus ${\Vdash_{\mathcal{B}} \bot, \Omega, \Theta_{\At}}$ by~\ref{eq:supp-bot}. \qed        
    \end{enumerate}

\end{proof}

\begin{theorem}[Completeness]\label{thm:completeness}
If $\Gamma \Vdash \Delta$ then the sequent $\Gamma \Rightarrow \Delta$ is provable in $\CLp$.
\end{theorem}

\begin{proof}
    Fix a simulation base $\ust$ with a mapping $\alpha$ for the set $\Gamma \cup \Delta$. Assume $\Gamma \Vdash \Delta$. Then $\Gamma \Vdash_{\ust} \Delta$. Now let $\bb$ be the system obtained by adding a rule concluding the sequent $\Rightarrow p^{B}$ from empty premises for all $B \in \Gamma$. So $\vdash_{\bb} \ \Rightarrow p^{B}$ for all $B \in \Gamma$, hence by Lemma \ref{mainlemma} we have $\Vdash_{\bb} B$ for all $B \in \Gamma$. From our assumption follows $\Vdash_{\bb} \Delta$, and from Lemma \ref{mainlemma} we once again have $\vdash_{\bb} \ \Rightarrow \Delta_{\At}$, so there must be a proof $\Pi$ in $\ust$ with conclusion $\Rightarrow \Sigma_{\At}$ for some $\Sigma_{\At} \subseteq \Delta_{\At}$.

    If $\Pi$ does not contain applications of the new rules added to $\bb$, it is already a deduction in the simulation base $\ust$, so by adding the context $\Gamma_{\At}$ to the left side of all sequents of the deduction we get a deduction $\Pi'$ with conclusion of shape $\Gamma_{\At} \Rightarrow \Delta_{\At}$. If it does contain applications of the new rules, we define a transformation similar to that of Lemma \ref{lemma:categoricalreducestohypothetical}. Let $\Pi'$ be the deduction obtained by adding to each application in $\Pi$ of one of the rules added to $\ust$ to obtain $\bb$ the context $\Gamma_{\At}$ on the left side of the sequent, as well as to all rule applications below them. The transformation can be represented graphically as follows:

    \bigskip
    \begin{center}
    
\begin{bprooftree}
    \AxiomC{}
    \UnaryInfC{$\Rightarrow p^{B}$}
    \noLine
    \UnaryInfC{$\Pi$}
    \noLine
    \UnaryInfC{$\Rightarrow \Delta_{\At}$}
\end{bprooftree}
is transformed into
\begin{bprooftree}
      \AxiomC{}
    \UnaryInfC{$\Gamma_{\At} \Rightarrow p^{B}$}
    \noLine
    \UnaryInfC{$\Pi'$}
    \noLine
    \UnaryInfC{$\Gamma_{\At} \Rightarrow \Delta_{\At}$}
\end{bprooftree}
\end{center}

    \bigskip

Since $p^{B} \in \Gamma_{\At}$ for every replaced rule due to the definition of $\bb$, clearly every instance of one of the new rules becomes an instance of atomic axiom already contained in $\ust$, so $\Pi'$ is a deduction with conclusion $\vdash_{\ust} \Gamma_{\At} \Rightarrow \Delta_{\At}$.

Hence in both cases we have a deduction $\Pi'$ using only the rules in $\ust$ with conclusion $\Gamma_{\At} \Rightarrow \Delta_{\At}$.  Now let $\Pi''$ be the deduction obtained by replacing every atom $p^{B}$ in the deduction $\Pi'$ by $B$. Since $\ust$ contains only instances of atomic axiom, atomic cut, and the rules of a simulation base, $\Pi''$ is a deduction in the classical sequent calculus (with cut) that ends with the sequent $\Gamma \Rightarrow \Delta$. Since the cut rule itself is admissible in $\CLp$ we conclude that there is a deduction in the classical sequent calculus that ends with the sequent $\Gamma \Rightarrow \Delta$. \qed
\end{proof}

\section{Cut-free completeness}
The proof of completeness just presented makes ostensive use of the cut rule. Since Gentzen's \textit{Haupsatz} shows that the cut rule is admissible in classical logic, one might be led to think that those uses are only for simplifying the proof, so completeness could also be proved through the same strategy if $\bu$ was an extension of $\hs$ instead of $\st$. However, this is not the case:

\begin{proposition}\label{failuremainlemma}
    Let $\Theta_{\At}$ be any (possibly empty) set of atoms, $\Sigma$ any (non-empty) set containing only subformulas of $\Gamma \cup \Delta$ and $\Sigma_{\At} = \{p^{B} | B \in \Sigma\}$. 
    \textbf{It is not the case} that, for all $\bb \supseteq \uhs$, $\Vdash_{\bb} \Sigma, \Theta_{\At}$ if and only if $\vdash_{\bb} \ \Rightarrow \Sigma_{\At}, \Theta_{\At}$.
\end{proposition}

\begin{proof}
As a quick counterexample, consider the consequence $q \land r \Vdash q$, which gives $\Gamma \cup \Delta = {q \land r, q}$. Let $\mathcal{B}$ be the extension of $\uhs$ obtained by adding an atomic axiom with conclusion $\Rightarrow p^{q \land r}$, where $q$ and $r$ are atomic formulas. By the definition of derivability (Definition~\ref{def:derivability}) we conclude $\vdash_{\bb} \Rightarrow p^{q \land r}$. Even though we can still apply the axiom and $L \land$ rules to conclude $\vdash_{\bb} p^{q \land r} \Rightarrow q$ and $\vdash_{\bb} p^{q \land r} \Rightarrow r$, the absence of cut on the base does not allow us to use those deductions in order to obtain the deductions showing $\vdash_{\mathcal{B}} \Rightarrow q$ and $\vdash_{\mathcal{B}} \Rightarrow r$ necessary to show $\Vdash_{\uhs} q \land r$, so we have to find other deductions. Notice, however, that no rule of the base is capable of producing a deductions with one of the desired conclusions; instances of axiom never have empty antecedents, and the only other rules in $\bb$ are the rules in $\uhs$ obtained via the mapping, that is, instances of $L \land$ with $p^{q \land r}$ on their antecedent and of $R \land$ with $p^{q \land r}$ in their consequent (given that $\Gamma \cup \Delta = \{q \land r, q\}$) . Remember that, in virtue of how the mapping is defined, $p^{q \land r} \neq p$ and $p^{q \land r} \neq p$. Since there are no deductions of $\Rightarrow p$ or $\Rightarrow q$ in $\bb$ we conclude $\nVdash_{\bb} p$ and $\nVdash_{\bb} q$, so clearly $\nVdash_{\bb} q \land r$. \qed
\end{proof}

The fact that cut is admissible in classical logic {\em does not entail} that is admissible in all atomic bases. This is essentially due to the fact that cut elimination procedures rely on the structure of rules capable of introducing a formula on the right of the sequent to permute cuts, but such rules can be arbitrary in bases. It is not possible in general to compose a deduction showing $ \Rightarrow p$ and a deduction showing $p \Rightarrow q$ to obtain one showing $ \Rightarrow q$ if the cut rule is not available, so by taking out such rules we change the semantic content of bases altogether.

This allows us to conclude that atomic cut rules indeed provide substantive semantic content to the bases containing them, but it does not clarify whether this content is essential for the proof of completeness to follow through. Notice that cut is not required in any step of the proof of soundness and, since $\st$ is an extension of $\hs$, completeness can be shown for $\cfv$ through the strategy employed in Theorem \ref{thm:completeness} if we require simulation bases $\uhs$ to also contain all instances of atomic cut, whence $\cfv$ is sound and complete w.r.t. classical logic. However, such a completeness proof would still rely on the use of atomic cuts.

As we will now show, completeness can still be proved without relying on bases closed under atomic cut. In the revised proof, all necessary applications of cut are deferred from Lemma~\ref{mainlemma} to the final step. This is allowed since, at that stage, we are working within a logical calculus rather than with arbitrary bases.

The importance of this proof lies on the fact that it replaces all the semantically meaningful (that is, non-admissible) applications of atomic cut by semantically redundant (admissible) applications of cut on a logical calculus. In other words, even though atomic cut definitely affect bases and the formulas they are capable of supporting, their semantic contents are made redundant at the level of logical connectives if a proof of cut admissibility is provided.

Instead of working with simulation bases, we use the same atomic mappings as before to define quasi-simulation bases $\bq$, which will enable us to prove a result similar to Lemma~\ref{mainlemma} without relying on applications of cut. Although replacing atoms $p^{A}$ with formulas $A$ in derivations from $\bq$ does not generally yield valid deductions in the classical sequent calculus, these derivations can be transformed into classical ones via suitable rule replacement procedures. This is possible since the new rules in $\bq$ implicitly incorporate the effect of cut into their structure.

\begin{figure}[t]
\[\infer[Q \land_{1}]{\Gamma_{\At} \Rightarrow \Delta_{\At}, p^{A}}{\Gamma_{\At} \Rightarrow \Delta_{\At}, p^{A \land B}}
\quad
\infer[Q \land_{2}]{\Gamma_{\At} \Rightarrow \Delta_{\At}, p^{B}}{\Gamma_{\At} \Rightarrow \Delta_{\At}, p^{A \land B}}
\quad
\infer[R \land]{\Gamma_{\At}, \Gamma_{\At}' \Rightarrow \Delta_{\At}, \Delta_{\At}', p^{A \land B}}{\Gamma_{\At} \Rightarrow  \Delta_{\At}, p^{A}& \Gamma_{\At}' \Rightarrow \Delta_{\At}', p^{B}}
\]
\\
\[
\infer[Q\lor]{\Gamma_{\At} \Rightarrow \Delta_{\At}, p^{A}, p^{B}}{ \Gamma_{\At} \Rightarrow  \Delta_{\At}, p^{A \lor B}}\quad
\infer[R \lor]{\Gamma_{\At} \Rightarrow \Delta_{\At}, p^{A \lor B}}{\Gamma_{\At} \Rightarrow \Delta_{\At}, p^{A}, p^{B}}
\]
\\
\[
\infer[Q \to]{\Gamma_{\At}, \Gamma_{\At}' \Rightarrow \Delta_{\At}, \Delta_{\At}', p^{B}}{\Gamma_{\At} \Rightarrow  \Delta_{\At}, p^{A \to B} & \Gamma_{\At}' \Rightarrow \Delta_{\At}', p^{A}}
\quad 
\infer[R \to]{\Gamma_{\At} \Rightarrow \Delta_{\At}, p^{A \to B}}{p^{A}, \Gamma_{\At} \Rightarrow \Delta_{\At}, p^{B}}
\quad
\infer[Q \bot]{\Gamma_{\At} \Rightarrow \Delta_{\At}}{{\Gamma_{\At} \Rightarrow \Delta_{\At}}, p^{\bot}}
\]
\caption{Rules of the quasi-simulation base $\bq$.}\label{fig:quasi}
\end{figure}

\begin{definition}[Quasi-simulation base]
Let $\Sigma$ be a set of formulas, and let $\alpha$ be an atomic mapping on $\Sigma$. A quasi-simulation base $\bq$ for $\Sigma$ and $\alpha$ is a base closed under $\hs$ that includes exactly the rules in Fig.~\ref{fig:quasi}, instantiated for all $A, B \in \Sigma$ and all multisets $\Gamma_{\At}, \Delta_{\At}, \Gamma_{\At}', \Delta_{\At}'$. 
\end{definition}
We proceed as follows.
\begin{lemma}\label{mainlemmacutfree}
    Let $\Theta_{\At}$ be any (possibly empty) set of atoms, $\Sigma$ any (non-empty) set containing only subformulas of $\Gamma \cup \Delta$ and $\Sigma_{\At} = \{p^{B} | B \in \Sigma\}$. 
    Then, for all $\bb \supseteq \bq$, $\Vdash_{\bb} \Sigma, \Theta_{\At}$ if and only if $\vdash_{\bb} \ \Rightarrow \Sigma_{\At}, \Theta_{\At}$.
\end{lemma}
\begin{proof}[of Lemma~\ref{mainlemmacutfree}]
We prove the result by induction on the complexity of $\Sigma \cup \Theta_{\At}$, understood as the number of logical connectives occurring in $\Sigma \cup \Theta_{\At}$.

Proof of the base case and left to right direction of all inductive steps are the same as in Lemma \ref{mainlemma}, as $\bq$ and its extensions still contain $R \land$, $R \lor$ and $R \to$. The right to left directions, on the other hand, are proved by using the new rules:

    \begin{enumerate}
        
        \item $\Sigma = A \land B \cup \Omega$. 

($\Leftarrow$): Assume $\vdash_{\bb} \ \Rightarrow p^{A \land B}, \Omega_{\At}, \Theta_{\At}$. First we do the following:
\begin{center}
\begin{bprooftree}
        \AxiomC{\vdots}
        \noLine
    \UnaryInfC{$\Rightarrow p^{A \land B},\Omega_{\At}, \Theta_{\At}$}
\RightLabel{\scriptsize{$Q \land_{1}$}}
\UnaryInfC{$\Rightarrow p^{A}, \Omega_{\At}, \Theta_{\At}$}
\end{bprooftree}

\begin{bprooftree}
        \AxiomC{\vdots}
        \noLine
    \UnaryInfC{$\Rightarrow p^{A \land B},\Omega_{\At}, \Theta_{\At}$}
\RightLabel{\scriptsize{$Q \land_{2}$}}
\UnaryInfC{$\Rightarrow p^{B}, \Omega_{\At}, \Theta_{\At}$}
\end{bprooftree}
\end{center}
This shows we have both $\vdash_{\bb} \ \Rightarrow p^{A}, \Omega_{\At}, \Theta_{\At}$ and $\vdash_{\bb} \ \Rightarrow p^{B}, \Omega_{\At}, \Theta_{\At}$. Induction hypothesis: $\Vdash_{\bb} A, \Omega, \Theta_{\At}$ and $\Vdash_{\bb} B, \Omega, \Theta_{\At}$. Then by~\ref{eq:supp-and} we get $\Vdash_{\bb} A \land B, \Omega, \Theta_{\At}$.
\item $\Sigma = A \lor B \cup \Omega$.

($\Leftarrow$): Assume $\vdash_{\bb} \ \Rightarrow p^{A \lor B}, \Omega_{\At}, \Theta_{\At}$. We do the following:
\begin{center}
\begin{bprooftree}
        \AxiomC{\vdots}
        \noLine
    \UnaryInfC{$\Rightarrow p^{A \lor B}, \Omega_{\At}, \Theta_{\At}$}
\RightLabel{\scriptsize{$Q \lor$}}
\UnaryInfC{$\Rightarrow p^{A}, p^{B}, \Omega_{\At}, \Theta_{\At}$}
\end{bprooftree}
\end{center}
So we conclude $\vdash_{\bb} \ \Rightarrow p^{A}, p^{B}, \Omega_{\At}, \Theta_{\At}$. Then the induction hypothesis yields $\Vdash_{\bb} A, B, \Omega, \Theta_{\At}$, so by~\ref{eq:supp-or}, $\Vdash_{\bb} A \lor B, \Omega, \Theta_{\At}$.
\item $\Sigma = A \to B \cup \Omega$

($\Leftarrow$):Assume $\vdash_{\bb} \ \Rightarrow p^{A \to B}, \Omega_{\At}, \Theta_{\At}$, and let $\bc$ be any $\bc\supseteq \bb$ with $\Vdash_{\bc} A, \Pi_{\At}$ for some set of atoms $\Pi_{\At}$. Then by induction hypothesis $\vdash_{\bc} \Rightarrow p^{A}, \Pi_{\At}$ and, since $\bc$ is an extension of $\bb$, we also have $\vdash_{\bc} \ \Rightarrow p^{A \to B},  \Omega_{\At}, \Theta_{\At}$. Then we can obtain the following deduction in $\bc$:
\begin{center}
\begin{bprooftree}
    \AxiomC{.}
    \noLine
    \UnaryInfC{.}
        \noLine
    \UnaryInfC{.}
        \noLine
    \UnaryInfC{$\Rightarrow p^{A \to B},  \Omega_{\At}, \Theta_{\At}$}
      \AxiomC{.}
    \noLine
    \UnaryInfC{.}
        \noLine
    \UnaryInfC{.}
        \noLine
    \UnaryInfC{$\Rightarrow p^{A}, \Pi_{\At}$}
    \RightLabel{\scriptsize{$Q \to$}}
\BinaryInfC{$\Rightarrow p^{B},  \Omega_{\At}, \Theta_{\At}, \Pi_{\At}$}
\end{bprooftree}
\end{center}
So $\vdash_{\bc} \ \Rightarrow p^{B}, \Omega_{\At}, \Theta_{\At}, \Pi_{\At}$, and by induction hypothesis $\Vdash_{\bc} B,\Omega, \Theta_{\At}, \Pi_{\At}$ (as $\Theta_{\At} \cup \Pi_{\At}$ is still a set of atoms). But $\bc$ is an arbitrary extension of $\bb$ with $\Vdash_{\bc} A, \Pi_{\At}$ for an arbitrary set of atoms $\Pi_{\At}$, so we conclude $A \Vdash_{\bb} B, \Omega, \Theta_{\At}$ by~\ref{eq:supp-inf} and then $\Vdash_{\bb} A \to B, \Omega, \Theta_{\At}$ by~\ref{eq:supp-imp}.
\item $\Sigma = \bot \cup \Omega$

($\Leftarrow$): Assume $\vdash_{\bb}\  \Rightarrow p^{\bot}, \Omega_{\At}, \Theta_{\At}$. We do the following:
\begin{center}
\begin{bprooftree}
      \AxiomC{.}
    \noLine
    \UnaryInfC{.}
        \noLine
    \UnaryInfC{.}
        \noLine
    \UnaryInfC{$\Rightarrow p^{\bot}, \Omega_{\At}, \Theta_{\At}$}
        \RightLabel{\scriptsize{$Q \bot$}}
    \UnaryInfC{$\Rightarrow \Omega_{\At}, \Theta_{\At}$}
\end{bprooftree}
\end{center}
So $\vdash_{\mathcal{B}} \Omega_{\At}, \Theta_{\At}$. The induction hypothesis yields $\Vdash_{\mathcal{B}} \Omega, \Theta_{\At}$, whence $\Vdash_{\mathcal{B}} \bot, \Omega, \Theta_{\At}$ by~\ref{eq:supp-bot}. \qed
    \end{enumerate}
\end{proof}

\begin{theorem}[Cut-free completeness]\label{thm:cutfreecompleteness}
$\Gamma \cfv \Delta$ implies that there is a classical proof of the sequent $\Gamma \Rightarrow \Delta$.
\end{theorem}

\begin{proof}
    Fix a quasi-simulation base $\bq$ with a mapping $\alpha$ for the set $\Gamma \cup \Delta$. Assume that $\Gamma \cfv \Delta$. Then it follows that $\Gamma \Vdash_{\bq} \Delta$. By applying the same procedure used in Theorem \ref{thm:completeness} we obtain a deduction $\Pi$ showing $\vdash_{\bq} \Gamma_{\At} \Rightarrow \Delta_{\At}$. Now let $\Pi'$ be the deduction obtained by replacing every atom $p^{B}$ in the deduction $\Pi$ by $B$. Denote by $Q\land^{*}_{1}$, $Q\land^{*}_{2}$, $Q\lor^{*}$, $Q\to^{*}$ and $Q\bot^{*}$ the rules obtained by replacing all atoms $p^{B}$ by formulas $B$  in the rules $Q\land_{1}$, $Q\land_{2}$, $Q\lor$, $Q\to$ and $Q\bot$, respectively. Now define the following transformation procedures:

\begin{prooftree}
    \AxiomC{\vdots}
    \noLine
    \UnaryInfC{$\Theta \Rightarrow A \land B, \Sigma$} \RightLabel{\scriptsize{$Q \land^{*}_{1}$}}
    \UnaryInfC{$\Theta \Rightarrow A, \Sigma$}
    \DisplayProof
    \quad
    $\mapsto$
    \quad
     \AxiomC{\vdots}
     \noLine
     \UnaryInfC{$\Theta \Rightarrow A \land B, \Sigma$}
    \AxiomC{}
    \RightLabel{\scriptsize{$init$}}
    \UnaryInfC{$A, B \Rightarrow A$}
    \RightLabel{\scriptsize{$\land L$}}
    \UnaryInfC{$A \land B \Rightarrow A$}
    \RightLabel{\scriptsize{$Cut$}}
    \BinaryInfC{$\Theta \Rightarrow A, \Sigma$}
\end{prooftree}

\begin{prooftree}
    \AxiomC{\vdots}
    \noLine
    \UnaryInfC{$\Theta \Rightarrow A \land B, \Sigma$} \RightLabel{\scriptsize{$Q \land^{*}_{2}$}}
    \UnaryInfC{$\Theta \Rightarrow B, \Sigma$}
    \DisplayProof
    \quad
    $\mapsto$
    \quad
     \AxiomC{\vdots}
     \noLine
     \UnaryInfC{$\Theta \Rightarrow A \land B, \Sigma$}
    \AxiomC{}
    \RightLabel{\scriptsize{$init$}}
    \UnaryInfC{$A, B \Rightarrow B$}
    \RightLabel{\scriptsize{$\land L$}}
    \UnaryInfC{$A \land B \Rightarrow B$}
    \RightLabel{\scriptsize{$Cut$}}
    \BinaryInfC{$\Theta \Rightarrow B, \Sigma$}
\end{prooftree}

\begin{prooftree}
    \AxiomC{\vdots}
    \noLine
    \UnaryInfC{$\Theta \Rightarrow A \lor B, \Sigma$} \RightLabel{\scriptsize{$Q \lor^{*}$}}
    \UnaryInfC{$\Theta \Rightarrow A,B,  \Sigma$}
    \DisplayProof
    $\mapsto$
     \AxiomC{\vdots}
     \noLine
     \UnaryInfC{$\Theta \Rightarrow A \lor B, \Sigma$}
    \AxiomC{}
    \RightLabel{\scriptsize{$init$}}
    \UnaryInfC{$A \Rightarrow A, \Sigma$}
     \AxiomC{}
    \RightLabel{\scriptsize{$init$}}
    \UnaryInfC{$B \Rightarrow B, \Sigma$}
    \RightLabel{\scriptsize{$\lor L$}}
    \BinaryInfC{$A \lor B \Rightarrow A, B, \Sigma$}
    \RightLabel{\scriptsize{$Cut$}}
    \BinaryInfC{$\Theta \Rightarrow A, B, \Sigma$}
\end{prooftree}

\begin{prooftree}
    \AxiomC{\vdots}
    \noLine
    \UnaryInfC{$\Theta \Rightarrow \bot, \Sigma$}
    \RightLabel{\scriptsize{$Q \bot^{*}$}}
    \UnaryInfC{$\Theta \Rightarrow \Sigma$}
    \DisplayProof
    \quad
    $\mapsto$
    \quad
    \AxiomC{\vdots}
    \noLine
    \UnaryInfC{$\Theta \Rightarrow \bot, \Sigma$}
    \AxiomC{}
    \RightLabel{\scriptsize{$\bot L$}}
    \UnaryInfC{$\bot \Rightarrow \varnothing$}
    \RightLabel{\scriptsize{$Cut$}}
    \BinaryInfC{$\Theta \Rightarrow \Sigma$}
\end{prooftree}

\begin{prooftree}
    \AxiomC{\vdots}
    \noLine
    \UnaryInfC{$\Theta \Rightarrow A \to B, \Sigma$} 
     \AxiomC{\vdots}
    \noLine
    \UnaryInfC{$\Theta' \Rightarrow A, \Sigma'$}
    \RightLabel{\scriptsize{$Q \to^{*}$}}
    \BinaryInfC{$\Theta, \Theta' \Rightarrow B,  \Sigma., \Sigma'$}
   \end{prooftree}
   \begin{prooftree}
   $\mapsto$
    \AxiomC{\vdots}
    \noLine
    \UnaryInfC{$\Theta \Rightarrow A \to B, \Sigma$} 
      \AxiomC{\vdots}
    \noLine
    \UnaryInfC{$\Theta' \Rightarrow A, \Sigma'$}
    \AxiomC{}
    \RightLabel{\scriptsize{$init$}}
    \UnaryInfC{$A \Rightarrow A$}
      \AxiomC{}
      \RightLabel{\scriptsize{$init$}}
    \UnaryInfC{$B \Rightarrow B$}
    \RightLabel{\scriptsize{$\to L$}}
    \BinaryInfC{$A \to B, A \Rightarrow B$}
    \RightLabel{\scriptsize{$Cut$}}
    \BinaryInfC{$\Theta', A \to B \Rightarrow B, \Sigma'$}
        \RightLabel{\scriptsize{$Cut$}}
    \BinaryInfC{$\Theta, \Theta' \Rightarrow B, \Sigma, \Sigma'$}
\end{prooftree}
\bigskip

Let $\Pi''$ be the deduction obtained by applying the procedures depicted above to all instances of $Q\land^{*}_{1}$, $Q\land^{*}_{2}$, $Q\lor^{*}$, $Q\to^{*}$ and $Q\bot^{*}$ in $\Pi'$. It is straightforward to check that $\Pi''$ is a deduction with conclusion $\Gamma \Rightarrow \Delta$ in the system obtained by adding the cut rule to our classical sequent calculus, hence since cut is admissible in that system we conclude that there is a deduction of $\Gamma \Rightarrow A$ in the classical sequent calculus.
\qed
\end{proof}

\section{Conclusion}
We have presented a novel formulation of base-extension semantics grounded in sequent calculus rather than natural deduction. This shift in perspective brings several benefits. First, the use of multiple-conclusion sequents enables a direct and elegant alignment with classical logic, overcoming traditional associations of proof-theoretic semantics with intuitionistic reasoning. By leveraging the invertibility and harmony of sequent rules we achieve semantic clauses that mirror classical inference patterns without resorting to external reasoning principles.

Our approach accommodates both cut-free and cut-inclusive systems, shedding light on the semantic role of the atomic cut rule. We demonstrated that atomic cuts significantly influence the behaviour of specific bases, even though they are not required for completeness at the level of logical connectives. This insight not only refines our understanding of the interplay between structural rules and semantics, but also reveals how semantically meaningful cuts can be modularly absorbed into logical reasoning.



Finally, our sequent-based $\Bes$ opens the door to broader applications. The modularity of the framework suggests clear pathways for extending the semantics to substructural and modal logics, as well as to other non-classical systems. For instance, even though the harmony of classical rules allows us to consider only the right side of the sequent when defining atomic support, a stronger clause with added structure could be defined as follows:

\[
\Vdash^{\Gamma_{\At}}_{\bb} \Delta_{\At}  \mbox{ iff } \Gamma_{\At} \Rightarrow \Delta_{\At} \mbox{ is provable in } \bb 
\]

This use of contexts is somewhat similar to the one observed in \cite{DBLP:journals/corr/abs-2402-01982,DBLP:conf/tableaux/GheorghiuGP23}, but now the contexts $\Gamma_{\At}$ are bound by the rules of the base instead of being arbitrary sets of hypotheses. It is also interesting to note that, if a stricter notion of derivability is used, atomic instances of weakening and contraction seem to have a concrete semantic impact on atomic bases, just like the instances of atomic cut rules in the current approach.

By reconstructing $\Pts$ in a sequent setting, we hope to contribute to a more general, flexible, and expressive theory of logical meaning---one that is well-positioned to interact with contemporary developments in logic, computation and philosophy.

\begin{credits}
\subsubsection{\ackname} Piotrovskaya is supported by the Engineering and Physical Sciences Research Council grants EP/T517793/1 and EP/W524335/1. Barroso-Nascimento and Pimentel are supported by the Leverhulme grant RPG-2024-196. Pimentel has received funding from the European Union's Horizon 2020 research and innovation programme under the Marie Sk\l odowska-Curie grant agreement Number 101007627.
We are grateful for the useful suggestions from the anonymous referees.
\end{credits}

\newpage
\bibliographystyle{splncs04}
\bibliography{references} 

\begin{thebibliography}{10}
\providecommand{\url}[1]{\texttt{#1}}
\providecommand{\urlprefix}{URL }
\providecommand{\doi}[1]{https://doi.org/#1}

\bibitem{DBLP:journals/corr/abs-2306-03656}
Barroso-Nascimento, V., Pereira, L.C., Pimentel, E.: An ecumenical view of
  proof-theoretic semantics. CoRR  \textbf{abs/2306.03656} (2023),
  \url{https://doi.org/10.48550/arXiv.2306.03656}

\bibitem{barrosonascimento2025prooftheoreticapproachsemanticsclassical}
Barroso-Nascimento, V., Piotrovskaya, E., Pimentel, E.: A proof-theoretic
  approach to the semantics of classical linear logic (2025),
  \url{https://arxiv.org/abs/2504.08349}, accepted to MFPS 2025.

\bibitem{DBLP:journals/corr/abs-2402-01982}
Buzoku, Y.: A proof-theoretic semantics for intuitionistic linear logic. CoRR
  \textbf{abs/2402.01982} (2024),
  \url{https://doi.org/10.48550/arXiv.2402.01982}, accepted to Studia Logica.

\bibitem{Sanz2016-SANODV}
de~Campos~Sanz, W., Oliveira, H.: On {D}ummett's verificationist justification
  procedure. Synthese  \textbf{193}(8),  2539--2559 (2016).
  \doi{10.1007/s11229-015-0865-3}

\bibitem{Cook10.1093/oxfordhb/9780195325928.003.0011}
Cook, R.: Intuitionism reconsidered. In: The Oxford Handbook of Philosophy of
  Mathematics and Logic. Oxford University Press (06 2007),
  \url{https://doi.org/10.1093/oxfordhb/9780195325928.003.0011}

\bibitem{Dummett1973-DUMFPO-2}
Dummett, M.: Frege: Philosophy of Language. Duckworth, London (1973)

\bibitem{dummett1991logical}
Dummett, M.: {The Logical Basis of Metaphysics}. Harvard University Press
  (1991)

\bibitem{DBLP:journals/jsyml/FerreiraF13}
Ferreira, F., Ferreira, G.: Atomic polymorphism. J. Symb. Log.  \textbf{78}(1),
   260--274 (2013), \url{https://doi.org/10.2178/jsl.7801180}

\bibitem{Gentzen1969}
Gentzen, G.: {Investigations into Logical Deduction}. In: Szabo, M.E. (ed.)
  {The Collected Papers of Gerhard Gentzen}. North-Holland Publishing Company
  (1969)

\bibitem{DBLP:journals/corr/abs-2503-05364}
Gheorghiu, A.V., Buzoku, Y.: Proof-theoretic semantics for classical
  propositional logic with assertion and denial. CoRR  \textbf{abs/2503.05364}
  (2025), \url{https://doi.org/10.48550/arXiv.2503.05364}

\bibitem{DBLP:conf/tableaux/GheorghiuGP23}
Gheorghiu, A.V., Gu, T., Pym, D.J.: Proof-theoretic semantics for
  intuitionistic multiplicative linear logic. In: Ramanayake, R., Urban, J.
  (eds.) Automated Reasoning with Analytic Tableaux and Related Methods - 32nd
  International Conference, {TABLEAUX} 2023, Prague, Czech Republic, September
  18-21, 2023, Proceedings. Lecture Notes in Computer Science, vol. 14278, pp.
  367--385. Springer (2023),
  \url{https://doi.org/10.1007/978-3-031-43513-3\_20}

\bibitem{DBLP:journals/logcom/HallnasS91}
Halln{\"{a}}s, L., Schroeder{-}Heister, P.: A proof-theoretic approach to logic
  programming. {II.} programs as definitions. J. Log. Comput.  \textbf{1}(5),
  635--660 (1991), \url{https://doi.org/10.1093/logcom/1.5.635}

\bibitem{DBLP:journals/igpl/Makinson14}
Makinson, D.: On an inferential semantics for classical logic. Log. J. {IGPL}
  \textbf{22}(1),  147--154 (2014). \doi{10.1093/jigpal/jzt038},
  \url{https://doi.org/10.1093/jigpal/jzt038}

\bibitem{DBLP:journals/apal/MarinMPV22}
Marin, S., Miller, D., Pimentel, E., Volpe, M.: From axioms to synthetic
  inference rules via focusing. Ann. Pure Appl. Log.  \textbf{173}(5),  103091
  (2022). \doi{10.1016/J.APAL.2022.103091},
  \url{https://doi.org/10.1016/j.apal.2022.103091}

\bibitem{LinearLogicND10.1093/jigpal/12.6.601}
Martins, L.R., Martins, A.T.: Natural deduction and weak normalization for full
  linear logic. Logic Journal of the IGPL  \textbf{12}(6),  601--625 (11 2004),
  \url{https://doi.org/10.1093/jigpal/12.6.601}

\bibitem{DBLP:journals/tocl/McDowellM02}
McDowell, R., Miller, D.: Reasoning with higher-order abstract syntax in a
  logical framework. {ACM} Trans. Comput. Log.  \textbf{3}(1),  80--136 (2002),
  \url{https://doi.org/10.1145/504077.504080}

\bibitem{DBLP:journals/tocl/MillerT05}
Miller, D., Tiu, A.: A proof theory for generic judgments. {ACM} Trans. Comput.
  Log.  \textbf{6}(4),  749--783 (2005),
  \url{https://doi.org/10.1145/1094622.1094628}

\bibitem{DBLP:journals/jphil/PiechaSS15}
Piecha, T., de~Campos~Sanz, W., Schroeder{-}Heister, P.: Failure of
  completeness in proof-theoretic semantics. J. Philos. Log.  \textbf{44}(3),
  321--335 (2015), \url{https://doi.org/10.1007/s10992-014-9322-x}

\bibitem{PRAWITZ1971235}
Prawitz, D.: Ideas and results in proof theory. In: Fenstad, J. (ed.)
  Proceedings of the Second Scandinavian Logic Symposium, Studies in Logic and
  the Foundations of Mathematics, vol.~63, pp. 235--307. Elsevier (1971),
  \url{https://www.sciencedirect.com/science/article/pii/S0049237X08708498}

\bibitem{Prawitz-2006}
Prawitz, D.: Meaning approached via proofs. Synthese  \textbf{148},  507--524
  (2006). \doi{10.1007/s11229-004-6295-2}

\bibitem{Sandqvist}
Sandqvist, T.: {Classical logic without bivalence}. Analysis  \textbf{69}(2),
  211--218 (04 2009). \doi{10.1093/analys/anp003}

\bibitem{Sandqvist2015IL}
Sandqvist, T.: Base-extension semantics for intuitionistic sentential logic.
  Logic J. of the IGPL  \textbf{23}(5),  719--731 (2015).
  \doi{10.1093/jigpal/jzv021}

\bibitem{pts-91}
Schroeder-Heister, P.: Uniform proof-theoretic semantics for logical constants
  (abstract). Journal of Symbolic Logic  \textbf{56}, ~1142 (1991)

\bibitem{sep-proof-theoretic-semantics}
Schroeder-Heister, P.: {Proof-Theoretic Semantics}. In: Zalta, E.N., Nodelman,
  U. (eds.) The {Stanford} Encyclopedia of Philosophy. Metaphysics Research
  Lab, Stanford University, {W}inter 2022 edn. (2022),
  \url{https://plato.stanford.edu/archives/win2022/entries/proof-theoretic-semantics/}

\end{thebibliography}
\end{document}